\documentclass[a4paper, 12pt]{article}
\usepackage[T1]{fontenc}
\usepackage[latin1]{inputenc}

\usepackage{amsfonts}
\usepackage{float}
\usepackage{graphicx}
\usepackage{setspace}
\usepackage{lscape}
\usepackage{color}
\usepackage{amssymb}
\usepackage{amsmath}
\usepackage{float}
\usepackage{xy}
\usepackage{subfigure}
\usepackage{setspace}
\usepackage{xcolor}
\usepackage[normalem]{ulem} 
\usepackage[colorlinks,citecolor=darkblue,linkcolor=maroon,bookmarks=true]{hyperref}
\usepackage[colorlinks,citecolor=darkblue,linkcolor=maroon,bookmarks=true]{hyperref}
\usepackage{tikz}

\usetikzlibrary{shapes, arrows.meta, positioning}

\usepackage{natbib}
\makeatletter
     {\def\@noitemerr
       {\@latex@warning{Empty `thebibliography' environment}}%
      \endlist}
\makeatother

\usepackage{amsmath, amsthm, amssymb, mathtools}
 \usepackage{graphicx, subfigure}

\usepackage{natbib}
\usepackage{color}

\newcommand{\noi}{\noindent}

\usepackage{setspace}

\usepackage{csquotes}

\usepackage{relsize}

\newcommand{\is}{$i$'s }

\newcommand{\sq}{\ensuremath{\subseteq}}

\newtheorem{Definition}{Definition}

\newtheorem{Assumption}{Assumption}

\newtheorem{Lemma}{Lemma}

\newtheorem{Example}{Example}

\newtheorem{Remark}{Remark}

\usepackage{eurosym}

\usepackage{chngcntr}
\counterwithout{figure}{section}

\usepackage{mathrsfs}
\usepackage[top=1.25in,bottom=1.25in,left=1.25in,right=1.25in]{geometry}

\definecolor{darkblue}{rgb}{0.0,0,.6}
\definecolor{maroon}{rgb}{0.68,0,0}
\definecolor{darkgreen}{rgb}{0,0.369,0.086}
\usepackage{xparse}
\usepackage{verbatim}

\usepackage[createShortEnv]{proof-at-the-end}

\newtheorem*{theorem*}{Theorem}

\newtheorem{proposition}{Proposition}

\usepackage{eurosym}

\begin{document}
\title{\vspace{-1.4cm} Information Aggregation with Costly Information Acquisition\vspace{0.6cm}%
\footnote{\scriptsize We would like to thank Emir Kamenica, Anastasios Karantounias, Ehud Lehrer, Suraj Malladi, Marciano Siniscalchi and participants at Hitotsubashi, Northwestern, DIMACS 2024 Workshop on Forecasting at Rutgers, the Workshop on Information Aggregation in Durham, the 10th LGT Conference in Lancaster, Athens University of Economics and Business, the 7th World Congress of the Game Theory Society in Beijing, the  EC24 Conference on Economics and Computation at Yale, Glasgow, Sussex, the CRETE Conference in Milos, Greece, the MIMA Workshop in Macroeconomic Theory at Warwick, the CEPR Workshop in Turin, the Durham York Workshop in Economic Theory, the  Durham Economic Theory Conference, EWMES2023 in Manchester, and SAET2023 in Paris. This research is funded under ESRC grant ES/V004425/1. 
}}

\author{\large Spyros Galanis\footnote{\scriptsize {Department of Economics, Durham University, Mill Hill Lane, Durham, UK, DH1 3LB. Email: spyros.galanis@durham.ac.uk. }}
\\ 
\and
\author({\large Sergei Mikhalishchev}\footnote{\scriptsize
{Department of Economics, Durham University, Mill Hill Lane, Durham, UK, DH1 3LB. Email: sergei.mikhalishchev@durham.ac.uk.}}}

\date{\small {\vspace{0.05in}\hspace{0.1cm}This draft: \today}}

\maketitle \thispagestyle{empty}

\vspace{-0.65cm}\begin{abstract}
We study information aggregation in a dynamic trading model with partially informed traders. \cite{ostrovsky12} showed that `separable' securities aggregate information in all equilibria, however,  determining whether a security is separable requires knowing the exact information structure of agents. To remedy this problem, we allow traders to acquire signals with cost $\kappa$, in every period. We show that `$\kappa$ separable securities' characterize information aggregation and, as the cost decreases, almost all securities become $\kappa$ separable, irrespective of the traders' initial private information.  Moreover, the switch to $\kappa$ separability happens not gradually but discontinuously, hence even a small decrease in costs can result in a security aggregating information. We provide a complete classification of securities in terms of how well they aggregate information, which surprisingly depends only on their payoff structure. 

 \vspace{0.6cm}

 \vspace{0.2cm}

\end{abstract}
\pagebreak \pagenumbering{arabic}

\section{Introduction}
The question of whether financial markets reveal and aggregate the private information of traders has been studied at least since \cite{hayek45}. \cite{ostrovsky12} provides a strong result, that information gets aggregated in {\it all Nash equilibria} if the traded securities are {\it separable} and trading takes place for infinitely many periods. A security is a function from states to payoffs. It is non-separable given an information structure if there is a prior under which its payoffs are uncertain, yet all traders agree on its expected value given their private information. If no such prior exists, it is separable. However, determining whether a security is separable requires knowing the exact information structure of agents. 
As a result, a market designer who does not know who participates or what is their information structure, cannot be sure that the equilibrium price is a good predictor of the security's value.

In this paper, we examine whether the ability to acquire costly signals during trading can make markets more efficient at aggregating information, so that the price of a security always converges to its value. This question becomes more relevant as the continuous improvements in information technology have created an abundance of available information, which is now cheaper than ever to acquire, analyze, and act upon.\footnote{For example, recent advances in generative AI tools such as ChatGPT could add considerable value for
investors with information processing constraints \citep{kim2023bloated} and assist in picking stocks \citep{pelster2023can}.}

Our main result, Theorem \ref{thm:info aggregation} with Proposition \ref{prop:different values}, is that for a generic security, which pays differently across all states, and for sufficiently low cost of information acquisition, information will get aggregated, {\it irrespective} of who trades and what is the information structure. In other words, the availability of cheap information makes `most' markets aggregate information under very general conditions.%
\footnote{This result is also supported by empirical evidence. \cite{farboodi2022has} show, using a structural model, that as the value of a firm's data grows, which is equivalent in our model to a decrease in the cost of information acquisition, the information content of the price of the firm's stock increases as well.
}

We use the dynamic trading model of \cite{ostrovsky12} with infinitely many periods and payoffs given by the Market Scoring Rule (MSR) \citep{hanson03,hanson07}. Each Trader's private information is represented by a partition of the state space, and the conjunction of everyone's private information reveals the true value of the security. The model is enhanced by allowing traders to acquire a signal ${\cal R}$ in each period, before trading, at a cost of $cK({\cal R})$,   where  $K$ is a cost function of signals and $c$ is a strictly positive marginal cost. Given $\kappa = (K,c)$, we say that a security is {\it $\kappa$ non-separable} given an information structure if it is non-separable at some prior and no trader finds it myopically optimal to acquire any signal at that prior. Otherwise, it is $\kappa$ separable.


Theorem \ref{thm:info aggregation} shows that $\kappa$ separable securities are necessary and sufficient for information aggregation, when the cost of information is $\kappa$, thus generalising \cite{ostrovsky12}. As the marginal cost $c$ decreases, the securities that eventually become $\kappa$ separable, and therefore aggregate information in {\it all} Nash equilibria and for {\it all} information structures, have a very simple structure: they specify a different payoff at each state (Proposition \ref{prop:different values}).
%

Surprisingly, there is a second, small class of securities, which Proposition \ref{prop:kappa non-separable characterization} characterises, that never become $\kappa$ separable for {\it all} information structures, and therefore may fail information aggregation, even when $c$ converges (but is not equal) to zero.%
\footnote{If the cost of information acquisition is zero, then information aggregates trivially.}
Such a security specifies the same payoff $d$ in two states, while in two other states it pays either higher or lower than $d$.

The third and final class consists of securities which are always separable, irrespective of who trades and what their information structure is. These securities aggregate information even if costly signals are unavailable. Unfortunately, this class is very small and uninformative. Proposition \ref{prop:always separable} shows that it only contains the Arrow-Debreu (A-D) security, which pays $a$ at some state and $b$ in all other states, and the security that specifies three payoffs: the largest is paid in one state of the world, the lowest in another, and the middle in all other states.%
\footnote{These securities are not very informative because the A-D can only predict whether one state has occurred or not, whereas the other security can only predict whether two states have occurred or not. Note that combining more than one A-D security to construct a composite and more informative security will not solve this issue, because it will be non-separable for some information structures.}
For any other security, there is a market (information structure and common prior) at which there is no information aggregation so that the security's price does not converge to its true value. This means that information may not aggregate in `most' markets if signals are too costly or unavailable.

These three classes provide a complete classification of securities in terms of information aggregation, which depends only on their payoff structure. First, the most informative class of $\kappa$ separable securities for some $\kappa$, ensure that information gets aggregated as long as the cost of information acquisition is low. Second, the securities which are $\kappa$ non-separable for all $\kappa$, may fail information aggregation even if the cost is low. Finally,  the always separable securities aggregate information even if costly signals are unavailable.

If the marginal cost of information acquisition gets arbitrarily small, do we even need markets to aggregate information through prices? Each Trader could buy the necessary signals and then bid very close to the true value. We argue that this intuition is incorrect, because markets become even more important in an environment with information acquisition. As the marginal cost $c$ decreases, the security switches discontinuously from $\kappa$ non-separable to $\kappa$ separable, hence even a slight reduction can enable a market to aggregate information in all Nash equilibria. This is in contrast to a poll, defined as the average of the traders' opinions after receiving the information, because its predictive accuracy improves smoothly as costs decrease (see Section \ref{sec: value of the market}). Hence, the availability of cheap information leverages the value of the markets, enabling them to aggregate information long before the cost goes to zero. In fact, Proposition \ref{prop: char k separable and poll} shows that a security is $\kappa$ separable if and only if the market is more accurate than a poll, for all priors.

\subsection{Literature}

%
Our paper contributes to two strands of the literature. The first studies the inefficiency of information acquisition and its effect on information aggregation in markets.%
\footnote{See  \cite{lim2011evolution} for a survey of the empirical literature on market efficiency.}
\cite{pavanEtAl22} show that traders acquire and use information inefficiently. Moreover, as the cost of information declines, traders over-invest
in information acquisition and trade too much on their private information. Several experimental studies
support these results and find that traders tend to over-acquire information. In addition, while information acquisition is positively correlated with market efficiency,  market prices do not aggregate all private information \citep{kraemer2006information, pageSiemroth17, page2021much, corgnet2022security}. \cite{MeleSangiorgi15} analyze costly information acquisition in asset markets with ambiguity-averse traders and show that when uncertainty is high enough, information acquisition decisions become strategic complements and lead to multiple equilibria. Our paper complements and differs from this literature. We find that, as the cost of information acquisition decreases, the number of securities (and therefore markets) that aggregate information increases. However, some securities are never able to aggregate information, even if the cost is almost zero. 


The second strand looks at the information aggregation properties of financial and, in particular, prediction markets.%
\footnote{See \cite{wolfersZitzewitz04} for an early overview of the literature.}
\cite{demarzoSkiadas98, demarzoSkiadas99} first introduced the notion of separable securities. \cite{ostrovsky12} and \cite{chenEtAl12} show that in a market with dynamically consistent traders, separable securities are both necessary and sufficient for information aggregation. \cite{dimitrovSami08} and \cite{chenEtAl10} examine information aggregation by varying the assumptions regarding the traders' information structure. \cite{galanisEtAl24} show theoretically and experimentally that ambiguity aversion can lead to no information aggregation with separable securities, but a new class - strongly separable securities - overcomes this limitation.
\cite{galanisKotronis21} show that information aggregation may fail if boundedly rational traders are unable to update their awareness when confronted with price announcements that they cannot rationalize.%
\footnote{Unawareness and ambiguity aversion generate dynamic inconsistency and negative value of information, which are partly responsible for no information aggregation. See \cite{gal11,gal13} for a model of unawareness and \cite{galanis21} for a connection between dynamic inconsistency and the negative value of information.}
We contribute to this literature by allowing traders to acquire costly signals at every period and propose a new class of $\kappa$ separable securities which characterize information aggregation.


We conclude by motivating our choice of the MSR model. First, in the MSR model, there are no noise traders and no strategic market makers, hence the issue of information aggregation is not intertwined with that of information revelation, as in \cite{kyle85}.%
\footnote{Note that \cite{ostrovsky12} uses both the MSR and the \cite{kyle85} model.}
Unlike the MSR, in \cite{kyle85} it is not always the case that the price will converge to the true value of the security, even if there is only one Trader and therefore information aggregation is achieved by default. Second, a prediction market with the MSR can be reinterpreted as an inventory-based market with a market maker who continuously adjusts the price of the securities depending on the orders she receives.%
\footnote{See \cite{ostrovsky12} and \cite{galanisEtAl24} for examples.}
The advantage of the MSR over more well-known market mechanisms, such as the continuous double auction, is that an agent can make her prediction/trade without waiting for another agent to take the opposite side, or submit a limit order and wait for it to be filled. This feature makes it an attractive mechanism for markets with relatively few participants who do not trade daily, or in markets with automated market makers.%
\footnote{Automated market makers are widely used in Decentralized Finance, see \cite{schlegelEtAl22} for an axiomatization of the logarithmic MSR. \cite{frongilloEtAl24} show the equivalence of prediction markets and constant function market makers that are used overwhelmingly when trading on the blockchain.}
MSR-based prediction markets have been used widely, for example, by firms such as Ford, Google, General Electric, and Chevron (see \cite{ostrovsky12}, \cite{CowgillZitzewitz15}) as well as governments, for example, in the UK and the Czech Republic (\cite{economist21}). 


The paper adheres to the following plan. Section \ref{model} describes the model and Section \ref{main result} presents our main results. In Section \ref{classification of securities} we provide a complete classification of securities. Section \ref{sec: value of the market} compares a poll with a prediction market and shows that a low cost of information acquisition leverages the value of the market. Section \ref{conclusion} concludes.

\section{The Model} 
\label{model}

\subsection{Preliminaries}
\label{preliminaries}
Uncertainty is described by a finite state space $\Omega= \{ \omega_{1},...,\omega_{l} \}$ and the set of traders is denoted $I = \{1, \ldots, n\}$. Trader \is initial private information is represented by partition $\Pi_{i}$ of $\Omega$. Let $\Pi_i(\omega)$ be a partition element of $\Pi_{i}$ that contains $\omega$, so that $\omega \in \Pi_i(\omega) \in \Pi_{i}$. When the true state is $\omega \in \Omega$, Trader $i$ considers all states in $\Pi_i(\omega) \sq \Omega$ to be possible. We assume that the join (the coarsest common refinement) of partitions $\Pi = \{\Pi_1, \ldots \Pi_n\}$ consists of singleton sets so that $\underset{i \in I}{\bigcap} \Pi_i(\omega) = \omega$ for all $\omega \in \Omega$, which means that the traders' pooled information always reveals the true state.%
\footnote{This assumption is also made by \cite{ostrovsky12} and it is without loss of generality because if the conjunction of the traders' private information does not reveal the state, we cannot expect that trading the security will reveal it.}
This implies that, for any two states $\omega_1 \neq \omega_2$, there exists Trader $i$ such that $\Pi_i(\omega_1) \neq \Pi_i(\omega_2)$. Let ${\cal P}$ be the collection of all information structures $\Pi$ where $\Omega$ has at least three states and $\underset{i \in I}{\bigcap} \Pi_i(\omega) = \{\omega\}$ for all $\omega \in \Omega$.  Traders have a full-support common prior $\mu_0$ over $\Omega$ and they are risk-neutral. 

\subsection{Trading environment}
Trading is organized as follows. There are infinitely many periods and finitely many $I$ traders, with $|I| = n$, who trade a security $X: \Omega \rightarrow {\mathbb R}$.  At time $t_0 = 0$, nature selects a state $\omega^* \in \Omega$ and the uninformed market maker makes a prediction $y_0$ about the value of security $X$. At time $t_1>t_0$, Trader 1 makes a revised prediction $y_1$, at $t_2>t_1$ Trader 2 makes his prediction, and so on. At time $t_{n+1} > t_n$, Trader 1 makes another prediction $y_{n+1}$, and the whole process repeats until time $t_{\infty} \equiv \lim_{k \to \infty}t_k =1$. 
All predictions are observed by all traders. Each prediction $y_k$ is required to be within the set $[\underline{y},\overline{y}] = [\underset{\omega\in \Omega}\min X(\omega), \underset{\omega\in \Omega}\max X(\omega)]$.
At some time $t^*>1$ the true value $x^* = X(\omega^*)$ of the security is revealed. 

The traders' payoffs from each announcement are computed using a scoring rule, $s(y,x^*)$, where $x^*$  is the true value of the security and $y$ is a prediction. A scoring rule is {\it proper} if, for any probability measure $\mu$ and any random variable $X$, the expectation of $s$  is maximized at $y=E_\mu[X]$. It is  {\it strictly proper} if $y$ is unique. We focus on continuous strictly proper scoring rules. Examples are the quadratic, where $s(y,x)=-(x-y)^{2}$, and the logarithmic, where $s(y,x)=(x-a)ln(y-a)+(b-x)ln(b-y)$ with $a<\underset{\omega\in \Omega}\min X(\omega), b>\underset{\omega\in \Omega}\max X(\omega)$. 

Under the market scoring rule (MSR) (\cite{mckelveyPage90}, \cite{hanson03, hanson07}), a Trader is paid for each revision he makes. In particular, his payoff, from announcing $y_n$ at $t_n$,   is $s(y_n,x^*)-s(y_{n-1},x^*)$, where  $y_{n-1}$ is the previous announcement and $x^*$ is the true value of the security. For all proper scoring rules, as $E_q[X]$ converges to $X(\omega)$, $s(E_q[X],X(\omega))$ converges to 0. Moreover, if $y_{n-1}$ is further away from $X(\omega)$ than $E_q(X)$ is from $X(\omega)$, then $s(E_q[X],X(\omega))-s(y_{n-1},X(\omega))$ is strictly positive. We then say that the Trader ``buys out'' the previous Trader's prediction. 
If he repeats the previous announcement, his period payoff is zero. 
%
%
We say that prior $\mu$ is non-degenerate given security $X$ if it does not assign probability 1 to a unique value of $X$. 

\subsection{Cost of Information}

We enhance the model of \cite{ostrovsky12} by allowing traders to acquire a costly signal in each time $t$ where they make an announcement. A signal (or statistical experiment) ${\cal R}$ consists of a signal realization space $T_{\cal R} \sq T$ and a family of distributions $\{{\cal R}(\cdot | \omega)\}_{\omega \in \Omega}$ over $T_{\cal R}$. 
Let ${\cal E}$ be the compact set of available signals, in the standard product topology on $\underset{{\cal R} \in {\cal E}}{\bigcup} \Delta (T_{\cal R})^\Omega$.%
\footnote{We use the Levy-Prokhorov metric $d_W$ (\cite{dekelEtAl06}) to generate the topology of weak convergence on $\Delta(T)$. The distance between two signals is defined as $D(R,P)= \underset{\omega \in \Omega}{\max} \; d_W ({\cal R}(\cdot | \omega), {\cal R}'(\cdot | \omega))$.}

Given a signal ${\cal R}$ and a prior $\mu$, each signal realization $\tau \in T_{\cal R}$ generates a posterior $\gamma \in \Delta(\Omega)$. Let ${\cal R}_\mu \in \Delta (\Delta (\Omega))$ be the generated random posterior, so that the probability of posterior $\gamma$ is ${\cal R}_\mu(\gamma)$.







\begin{Definition} \label{def:cost experiments}
The cost of signal ${\cal R}$ is $c K({\cal R})$,
where $c>0$ is the marginal cost and $K: {\cal E} \rightarrow [0, +\infty]$ is a continuous function.
Let ${\cal K}$ be the collection of all cost structures $\kappa = (c, K)$,
where $K$ satisfies the following assumptions:

    \begin{enumerate}
    \item[(i)] If $\kappa = (c, K) \in {\cal E}$ then $\kappa' = (c', K) \in {\cal E}$, for all $c'   >0$,
    \item[(ii)] There is an uninformative signal ${\cal R}_0$, so that ${\cal R}_0(\cdot|\omega) = {\cal R}(\cdot|\omega')$ for all $\omega, \omega' \in \Omega$, with $K({\cal R}_0) = 0$, and $K$ is sufficiently smooth around ${\cal R}_0$,%
    \footnote{That is, $\underset{{\cal R} \rightarrow {\cal R}_0}{\lim \sup} \frac{K({\cal R})}{[D({\cal R}, {\cal R}_0)]^2} < \infty$ for all ${\cal R} \in {\cal E}$.}
    \item[(iii)] If ${\cal R}'$ is a garbling of ${\cal R}$ then $K({\cal R}) \geq K({\cal R}')$.%
            \footnote{Recall that signal ${\cal R}'$ is a garbling of signal ${\cal R}$ if there exists a map $g: T_{\cal R} \rightarrow \Delta (T_{{\cal R}'})$ such that for each $\tau' \in T_{{\cal R}'}$,  ${\cal R}'(\tau'|\omega) = \underset{\tau \in T_{{\cal R}}}{\sum} g(\tau'|\tau) {\cal R}(\tau|\omega)$, for all $\omega \in \Omega$.}


    \end{enumerate}

\end{Definition}

We separately specify the marginal cost so that, for any $K$, we can express a sufficiently low cost of information acquisition with $c \rightarrow 0$, and (i) specifies that all marginal costs are available.
Condition (ii)  says that an uninformative signal is available and free, and $K$ is sufficiently smooth around zero information. Finally, (iii) says that more informative signals are costlier. 




For the second part of Theorem \ref{thm:info aggregation} and Proposition \ref{prop: char k separable and poll}, we assume that if a Trader can generate a random posterior $Q$ with a sequence of signals, it is weakly less costly to generate it with just one signal. This assumption is similar to the Sequential Learning-Proof property in \cite{bloedelZhong24}. In particular, consider a prior $\mu$, a signal ${\cal R}$, and a collection of  signals $\{{\cal R}^\tau\}_{\tau \in {T_{\cal R}}}$, one for each signal realization $\tau \in T_{\cal R}$. Let $Q = \{{\cal R} \otimes \{{\cal R}^\tau\}_{\tau \in T_{\cal R}}\}_\mu$ be the random posterior generated by the combination of these signals, according to the following procedure. First, the Trader conducts signal ${\cal R}$. After receiving each signal realization $\tau \in T_{\cal R}$ and updating his posterior to $\gamma^\tau$, he conducts signal ${\cal R}^\tau$. After receiving a new signal realization  $\tau' \in T_{{\cal R}^\tau}$, he updates his posterior to $\gamma^{\tau'}_{{\cal R}^\tau}$. The ex ante cost of the generated random posterior is $C = c \left (K({\cal R}) + \underset{\tau \in T_{\cal R}}{\sum} {\cal R}_\mu (\gamma^\tau) K({\cal R}^\tau) \right )$. We assume that the same random posterior $Q = \{{\cal R} \otimes \{{\cal R}^\tau\}_{\tau \in T_{\cal R}}\}_\mu$ can be generated by just one signal ${\cal R}'$, which costs weakly less than $C$.

\begin{Assumption} \label{ass: costly sequence of signals}
    For any prior $\mu$ and signals ${\cal R}$, $\{{\cal R}^\tau\}_{\tau \in T_{\cal R}}$, there is signal ${\cal R}'$ such that $\{{\cal R} \otimes \{{\cal R}^\tau\}_{\tau \in T_{\cal R}}\}_\mu = {\cal R'}_\mu$ and  $K({\cal R}) + \underset{\tau \in T_{\cal R}}{\sum} {\cal R}_\mu (\gamma^\tau) K({\cal R}^\tau) \geq K({\cal R}')$. 
\end{Assumption}



\subsection{The Game}

A game $\Gamma(\Omega, \Pi, X, \mu, {\cal E}, c, K, y_0, s, \beta)$ is a tuple where $\Omega$ is the state space, $\Pi$ is the collection of partitions which describes the initial private information of traders, $X$ is the traded security, $\mu$ is the common prior, ${\cal E}$ is the collection of signals, $c$ is the marginal cost and $K$ is the cost function of signals,  $y_{0}$ is the market maker's initial announcement at time $t_0$, $s$ is a strictly proper scoring rule, and $\beta$ is the common discount rate. The range of possible announcements is $[\underline{y},\overline{y}] = [\underset{\omega \in \Omega}{\min}X(\omega), \underset{\omega \in \Omega}{\max}X(\omega)]$.

At time $t_k$, one Trader chooses a signal and specifies an announcement in $[\underline{y},\overline{y}]$, for each possible signal realization. Then, nature moves and picks a signal realization. Given the Trader's specification, the announcement $y_k \in [\underline{y},\overline{y}]$ corresponding to that signal realization is made. The announcement $y_k$ becomes common knowledge but the choice of the signal and the signal realization are the private information of the Trader who announced.

Following  \cite{dimitrovSami08}, the payoff of Trader $i$  who makes an announcement $y_k$ at time $t_k$ is $\beta^k \left (s(y_{t_k},x^*)-s(y_{t_{k-1}},x^*) - cK({\cal R})\right)$, where $0 < \beta \leq 1$ is the common discount factor and ${\cal R} \in {\cal E}$ is the signal that he acquired.%
\footnote{For some examples, we allow myopic traders with $\beta=0$, where each Trader makes the myopically best announcement in each period.}
Both the choice of signal ${\cal R}$ and the signal realization is private and not revealed to other traders. The total payoff of each Trader is the sum of all payments for revisions, minus the cost of the signals he acquired. We also consider the special case of myopic, or non-strategic, traders, so that when making an announcement they do not take into consideration the impact of the other traders' strategies on their future payoffs. Note that each Trader can guarantee a zero payoff at any round, by repeating the previous announcement and not buying any signal.

Let $(y_1, \ldots, y_k) \in [\underline{y}, \overline{y}]^k$ be a history of announcements 
and $(\tau_1, \ldots, \tau_k) \in (\underset{{\cal R} \in {\cal E}}{\bigcup} T_{\cal R})^k$ be a history of signal realisations, up to time $t_k$. Denote by $T^k_i$ the collection of all histories about $i$'s signal realisations up to time $t_k$. For simplicity, we assume  that a signal realisation $\tau \in T_{\cal R}$ is unique to a specific signal, therefore a history of realisations in $T^k_i$ also identifies the chosen signals.



Let $[\underline{y}, \overline{y}]^{T_{\cal R}}$ be the set of all functions from signal realizations of ${\cal R}$ to announcements in $[\underline{y}, \overline{y}]$.
Trader \is mixed strategy at time $t_k$ is a measurable function 



\[\sigma_{i,k}:\Pi_i \times T^{k-1}_i \times [\underline{y}, \overline{y}]^{k-1} \times [0,1] \longrightarrow \underset{{\cal R} \in {\cal E}}{\bigcup}[\underline{y}, \overline{y}]^{T_{\cal R}}.\]

It specifies a signal ${\cal R}$ and an announcement $y_k$ for each realization that is drawn from that signal, given the element of his partition at the true state $\omega$, the history of his past chosen signal realizations, the history of all announcements up to time $t_{k-1}$, and the realization of random variable $\iota_{k}\in [0,1]$, which is drawn from the uniform distribution. One such draw takes place at each time $t_k$, it is observed only by the Trader who makes the announcement at $t_k$, and the draws are independent of each other and of the true state $\omega$.%
\footnote{This formulation, with draws on $[0,1]$ denoting the Trader's randomization, is taken from \cite{ostrovsky12}.}
The draws of the uniform distribution allow each Trader to randomize both in terms of the chosen signals and the announcements given the signal realisations. Note that, although all announcements are public, the history of past signals and realisations for each Trader $i$ is private.


The full state $\phi = (\omega, \iota_1, \iota_2, \ldots)$ resolves the initial uncertainty about $\Omega$ and the randomizations of the players according to $\iota$. Let $X(\phi) \equiv X(\omega)$ be the true value of $X$ at $\phi = (\omega, \iota_1, \iota_2, \ldots)$ and $\Phi = \Omega \times [0,1]^{\mathbb N}$ be the full state space. Denote by $\sigma_i$ the collection $\sigma_{i,k}$ of \is strategies, at all times $t_k$ where it is his turn to make an announcement. Let $\sigma = (\sigma_1, \ldots, \sigma_n)$ be a profile of strategies and $\Sigma$ be the set of all strategy profiles. Given a strategy $\sigma$ and state $\phi$, let  $y_{i+nk}(\sigma,\phi)$ be the announcement of Trader $i$ and $cK({\cal R}_{i+nk}(\sigma,\phi))$ the cost of acquiring signal ${\cal R}_{i+nk}(\sigma,\phi)$ at time $t_{i+nk}$.

\begin{Definition}\label{nash equilibrium}
A strategy profile $\sigma$ is a Nash equilibrium if, for every Trader $i$ and every alternative strategy $\sigma'= (\sigma_{-i},\sigma'_i)$, we have
\[E_\mu\Bigg[\sum_{k=0}^{\infty}\beta^{i+n k}\bigg(s \Big (y_{i+nk} (\sigma, \phi), X(\phi) \Big) - s \Big (y_{i+ nk-1} (\sigma, \phi), X(\phi) \Big) -cK \Big ({\cal R}_{i+nk}(\sigma,\phi) \Big) \bigg) \Bigg]\geq \]
\[E_\mu\Bigg[\sum_{k=0}^{\infty}\beta^{i+n k}\bigg(s \Big (y_{i+nk} (\sigma', \phi), X(\phi) \Big) - s \Big (y_{i+ nk-1} (\sigma', \phi), X(\phi) \Big) -cK \Big({\cal R}_{i+nk}(\sigma',\phi) \Big) \bigg) \Bigg],\]
where the expectation is taken with respect to the common prior $\mu$.
\end{Definition}


\subsection{Information Aggregation}
We say that information is aggregated if the traders' predictions converge to the intrinsic value $X(\phi)$ of security $X$, for all $\phi \in \Phi$. For every $\phi \in \Phi$ and strategy profile $\sigma$, let $y_k(\sigma, \phi)$ be the announcement of the Trader who moves in period $t_k$. The announcement $y_k(\sigma, \phi)$ depends on $\phi = (\omega, \iota_1, \iota_2, \ldots)$ because traders have different private information across $\Omega$ and randomize based on the realizations of random variable $\iota$. Because $\{y_k\}_{k=1}^{\infty}$  is a sequence of random variables, we need a probabilistic version of convergence.


\begin{Definition} \label{def info agg}
Under a profile of strategies in $\Gamma$, information aggregates if sequence $\{y_k\}_{k=1}^{\infty}$ converges in probability to random variable $X$.
\end{Definition}

\subsection{Separable Securities}
Consider the following example, which illustrates the notion of separability.

\begin{Example} \label{main example}
The state space is $\Omega = \{\omega_1, \omega_2, \omega_3, \omega_4\}$, the security is $X= (0,1,2,3)$, and there are two myopic traders with common prior $\mu=(\frac{1}{8},\frac{3}{8},\frac{3}{8},\frac{1}{8})$. Trader 1's partition is $\{\{\omega_1, \omega_3\}, \{\omega_2, \omega_4\}\}$ and Trader 2's is $\{\{\omega_1, \omega_4\}, \{\omega_2, \omega_3\}\}$.
\end{Example}

As traders are myopic in this example, under the MSR they take turns in announcing their expected value of $X$. Suppose that the true state is $\omega_1$. In period 1, Trader 1 with conditional beliefs $(\frac{1}{4}, 0, \frac{3}{4}, 0)$ will announce $\frac{3}{2}$. The same announcement would be made by Trader 1 in all states, and thus no information is transmitted to Trader 2. In period 2, Trader 2 also makes the same announcement at $\omega_1$, and the same announcement would be made in all states. Hence, no information is transmitted to Trader 1. Because the two traders agree on the announcement, there is no information updating, and the process ends. We say that there is no information aggregation at $\omega_1$ because the final announcement $\frac{3}{2}$ is not equal to the intrinsic value of $X$ at $\omega_1$, which is 0.%
\footnote{Information aggregation fails in the first round in this example and no-one updates from $\mu$. In general, traders could start from a different common prior $\mu'$, and after several rounds of updating they could update to some other posterior $\mu''$, at which there is no information aggregation.}

At all states, it is common knowledge that both traders agree that the expected value of $X$ is $\frac{3}{2}$, hence there can be no more learning from further (myopic) announcements.
However, it is also common knowledge that the intrinsic value of $X$ is not $\frac{3}{2}$; it is either 0, 1, 2 or 3, which implies that there is no information aggregation. When both these conditions are satisfied for some prior and some partitions, the security is non-separable. If there is no such prior, the security is separable.

\begin{Definition}
\label{separable}
A security $X$ is called non-separable under information structure $\Pi$ if there exists probability distribution $\mu$ and value $v\in \mathbb{R}$ such that:
\begin{itemize}
\item[$(i)$] $X(\omega) \neq v$ for some $\omega \in Supp(\mu)$,
\item[$(ii)$] $E_{\mu}[X|\Pi_{i}(\omega)]=v$ for all $i=1,...,n$ and $\omega \in Supp(\mu)$. 
\end{itemize}
We then say that security $X$ is non-separable at $\mu$. Otherwise, it is called separable.
\end{Definition}

A security $X$ is non-separable if, for some prior $\mu$ and at all states in its support, all traders' expected value of $X$ is $v$, yet there is uncertainty about the value of $X$. If for all priors at least one condition is violated, then the security is separable. Note that separability (and non-separability) is a property that depends on the information structure $\Pi$. For a different $\Pi$, a security may switch from separable to non-separable and vice versa. Theorem 1 in \cite{ostrovsky12} shows that separability characterizes information aggregation, in an environment without information acquisition.

\subsection{Information Acquisition}

To illustrate how information gets aggregated, even when the security is non-separable, we enhance Example \ref{main example} by allowing traders to acquire extra information about the value of the security through noisy signals, before making their announcement. Suppose that Trader 1 can acquire signal ${\cal R}$ that generates a signal realization $z$, given each state, with probabilities ${\cal R}(z | \omega_1)=1$ and ${\cal R}(z | \omega_2)={\cal R}(z | \omega_3)={\cal R}(z | \omega_4)=0.5$.

At state $\omega_1$, Trader 1's prior belief is $(\frac{1}{4}, 0, \frac{3}{4}, 0)$  and, after receiving the signal realization $z$, his posterior belief is $(\frac{2}{5},0, \frac{3}{5}, 0)$. Then, he announces the expected value of $X$, which is $\frac{6}{5}$. Trader 2 considers states $\omega_1$ and $\omega_4$ to be possible. The public announcement of $\frac{6}{5}$ reveals to Trader 2 that the true state is $\omega_1$. The reason is that, regardless of whether Trader 1 received a signal $z$ or not, his posteriors at $\omega_4$ are $(0, \frac{3}{4},0,\frac{1}{4})$ and he would have announced $\frac{3}{2}$. As a result, Trader 2 announces 0 in the second round, and the game ends. Note that the final price is equal to the intrinsic value of $X$ at $\omega_1$, hence information gets aggregated. In summary, the ability to acquire extra signals transforms $X$ from non-separable to $\kappa$ separable, enabling information aggregation.

We make two observations. First, we have abstracted from the cost of acquiring an experiment. Each Trader will acquire the signal only if the expected continuation value from making a better prediction outweighs the cost. 
Second, Trader 2 free-rides on Trader 1 buying the signal. By moving the price from $\frac{6}{5}$ to $0$, he books a profit, without paying the cost of a signal. 

\subsection{$\kappa$ Separable Securities}

Example \ref{main example} showed that even if a security is non-separable at some prior $\mu$, as long as traders can buy costly signals they can further move the announcement and achieve information aggregation. Hence, information aggregation may get `stuck' only if no-one is willing to acquire any new signals. We say that there is no (myopic) information acquisition if no Trader is willing to acquire any signal, given the cost structure $\kappa = (K,c)$.%
\footnote{It is important to emphasize that even though we use a `myopic' argument to define no information acquisition, traders are strategic in our setting.}


\begin{Definition}
There is no information acquisition for security $X$, given prior $\mu$, cost structure $\kappa = (K,c)$  and information structure $\Pi \in {\cal P}$, if for all states $\omega \in Supp(\mu)$,  all traders $i$, and all signals ${\cal R}$,
\begin{equation*}
\underset{\gamma \in Supp({\cal R}_{\mu_{\Pi_i(\omega)}})}{\sum} {\cal R}_{\mu_{\Pi_i(\omega)}}(\gamma) \underset{\omega' \in \Omega}{\sum} \gamma(\omega') \Big [s(E_{\gamma}[X],X(\omega'))-s(E_{\mu_{\Pi_i(\omega)}}[X],X(\omega')) \Big ]-cK({\cal R})\leq 0.
\end{equation*}
Otherwise, there is information acquisition.
\end{Definition}

To explain this inequality, suppose that at time $t$ it is Trader \is turn to make an announcement. Having observed all previous announcements and using the public information that is revealed, an outside observer updates the common prior $\mu_0$ to a belief $\mu$ over $\Omega$. If the true state is $\omega$, then Trader \is private information is $\Pi_i(\omega)$ and his posterior belief is the Bayesian update of $\mu$, denoted $\mu_{\Pi_i(\omega)}$. In other words, he updates using both his private information and the public information revealed by previous announcements.

His myopic problem at $\omega$ consists of buying an experiment ${\cal R}$, at cost $cK({\cal R})$, which generates a random posterior ${\cal R}_{\mu_{\Pi_i(\omega)}}$.  For each posterior belief $\gamma \in Supp({\cal R}_{\mu_{\Pi_i(\omega)}})$, his myopic best response is to announce $E_\gamma [X]$ because the scoring rule is proper. If his expected utility is weakly negative for all signals ${\cal R}$, then we say that he does not want to acquire any information, because buying no signal and repeating the previous announcement $E_{\mu_{\Pi_i(\omega)}}[X]$ guarantees a payoff of zero.%
\footnote{Note that, as in \cite{ostrovsky12}, the previous announcement does not influence the choice of the myopically optimal choice of the announcement (and the signal). We have specified $E_{\mu_{\Pi_i(\omega)}}[X]$ to be the previous announcement so that we normalize at zero the utility from no information acquisition, because then the optimal myopic announcement is $E_{\mu_{\Pi_i(\omega)}}[X]$. Interestingly,  the myopic best depends on the previous announcement when traders are ambiguity averse, as shown in \cite{galanisEtAl24}.}

We say that security $X$ is $\kappa$ non-separable given information structure $\Pi$ and cost structure $\kappa$ if there is a prior $\mu$ such that there is no myopic information acquisition and the security is non-separable at $\mu$.

\begin{Definition}
Security $X$ is $\kappa$ non-separable given an information structure $\Pi \in {\cal P}$ and cost structure  $\kappa \in {\cal K}$, if there exists prior $\mu$ such that
\begin{itemize}
\item Security $X$ is non-separable at $\mu$,
\item There is no information acquisition for $X$ given $\mu$.

\end{itemize}
Otherwise, security X is $\kappa$ separable.
\end{Definition}

\section{Main Result}
\label{main result}

Our main result consists of two parts, the first bullet point of Theorem \ref{thm:info aggregation} and Proposition \ref{prop:different values}. Collectively, they specify that for a security which pays differently across all states, and for sufficiently low cost of information acquisition, information will get aggregated, {\it irrespective} of who trades and what is the information structure $\Pi$.

The first part, Theorem \ref{thm:info aggregation}, shows that $\kappa$ separable securities characterize information aggregation when the cost structure is $\kappa = (K,c)$. 

\begin{theoremEnd}{theorem}\label{thm:info aggregation}
    Fix information structure $\Pi$ and cost structure $\kappa = (K,c)$. Then:
    \begin{itemize}
        \item If security $X$ is $\kappa$ separable under $\Pi$, then information gets aggregated in any Nash equilibrium of  game $\Gamma(\Omega, \Pi, X, \mu, {\cal E}, c, K, y_0, s, \beta)$, 
        \item If security $X$ is $\kappa$ non-separable under $\Pi$, then there exists prior $\mu$ such that for all $s$, $y_0$, and $\beta$,  there exists a Nash equilibrium of the corresponding game $\Gamma$ in which information does not get aggregated.
    \end{itemize}
\end{theoremEnd}

\begin{proofEnd}
    
We follow the proof of Theorem 1 in \cite{ostrovsky12}, which proceeds in four steps. In the first step, we show that there is a uniform lower bound on the expected profits that at least one Trader can make by improving the forecast.


Let $r$ be a distribution over $\Omega$ and $z$ be the previous announcement. Following \cite{ostrovsky12}, we define the instant opportunity of Trader $i$ as the highest expected payoff that he can achieve by changing the forecast from $z$, if the state is drawn according to $r$. The difference from \cite{ostrovsky12} is that we allow the Trader to acquire information before making an announcement. Let ${\cal R}_\omega$ be the optimal signal that Trader $i$ acquires at $\Pi_i(\omega)$, given his beliefs $r_{\Pi_i(\omega)}$, and $Q_\omega$ the resulting random posterior. Trader \is instant opportunity is
\[\underset{\omega \in \Omega}{\sum} r(\omega) \left ( \underset{\gamma \in Supp(Q_\omega)}{\sum} Q_\omega(\gamma) \underset{\omega' \in \Omega}{\sum} \gamma(\omega') \Big [s(E_{\gamma}[X],X(\omega'))-s(z,X(\omega')) \Big ]-cK({\cal R}_\omega) \right ).\]

Let $\Delta$ be the set of probability distributions $r \in \Delta(\Omega)$ for which there is uncertainty about security $X$, so that there are states $a,b$ with $r(a)>0$, $r(b)>0$ and $X(a) \neq X(b)$.

\begin{Lemma} \label{lemma 1 ostrovsky}
If security $X$ is $\kappa$ separable, then for every $r \in \Delta$, there exist $\chi > 0$ and Trader $i$ such that, for every $z \in {\mathbb R}$, the instant opportunity of Trader $i$ given $r$ and $z$ is greater than $\chi$. 
\end{Lemma}

\begin{proof}

Fix $r \in \Delta$. There are two cases. First, $X$ is separable with respect to $r$. Then, the proof of Lemma 1 in \cite{ostrovsky12} applies and we have the result. Second, $X$ is non-separable at $r$ and value $v$, where $v = E_r[X|\Pi_i(\omega)]$ for all $\omega \in Supp(r)$ and $i \in I$. Because $X$ is $\kappa$ separable, there is information acquisition. This implies that for some $\omega \in Supp(r)$, Trader $i$ receives a strictly positive payoff by acquiring information and changing the previous announcement $v = E_r[X|\Pi_i(\omega)]$. Note that the new announcement is not deterministic but depends on the optimal random posterior $Q$, so he announces $E_\gamma[X]$ with probability $Q(\gamma)$. From continuity, there is a small enough $\epsilon > 0$, so that for all previous announcements $z \in [v-\epsilon, v+\epsilon]$, Trader $i$ receives a strictly positive payoff of at least $\chi_1 > 0$ by acquiring information and changing the announcement. If $z \notin [v-\epsilon, v+\epsilon]$, then Trader $i$ receives a strictly positive payoff of at least $\chi_2>0$, by not acquiring any information and announcing the myopically best $E_r[X|\Pi_i(\omega)]=v$.%
\footnote{ \label{fn: ostrovsky p2618} This is true because a proper scoring rule is `order-sensitive' so that the further away the previous announcement is from the true expected value $E_r[X|\Pi_i(\omega)]=v$, the higher the payoff is for Trader $i$. The lowest payoff from the myopically best announcement is 0 and it is achieved when it is equal to the previous announcement (see p. 2618 in \cite{ostrovsky12}).}
By setting $\chi_3 = \min\{\chi_1, \chi_2\}$, we have that at $\Pi_i(\omega)$, Trader $i$ receives a strictly positive payoff of at least $\chi_3$, for all previous announcements $z$. For any $\omega' \in Supp(r) \setminus \Pi_i(\omega)$, Trader $i$ can repeat $v$ and receive 0. Hence, the instant opportunity of $i$ given $r$, which is the ex-ante expectation over all $\omega \in Supp(r)$, is at least $\chi = r(\Pi_i(\omega)) \chi_3 >0$ for all previous announcements.

\end{proof}

Steps 2-3 are identical to the proof of \cite{ostrovsky12} and establish that there is Trader $i$ and lower bound $\nu^*>0$ such that, for infinitely many periods $t_{nk +i}$,  the expected instant opportunity of Trader $i$ is greater than $\nu^*$. The reason they are identical is that, once we fix an equilibrium strategy, the only difference from \cite{ostrovsky12} is that traders buy an extra, payoff-irrelevant, signal in each period where they make an announcement. 

In Step 4, we show that the presence of a ``non-vanishing arbitrage opportunity'' is impossible in equilibrium. The proof is very similar to that of \cite{ostrovsky12}, but we write it here for completeness.

Let $\overline{s}_k$ be the expected score of prediction $y_k$, where the expectation is over all $\phi$ and the moves of players according to the mixed equilibrium. The expected payoff to the Trader who moves in period $t_k$ is $\beta^k(\overline{s}_k-\overline{s}_{k-1} - \overline{c}_k)$, where $\overline{c}_k \geq 0$ is the expected cost of information acquisition. It is zero if and only if there is no information acquisition.

Take any period $t_k$. Let $\Psi_k$ be the sum of all players' expected continuation payoffs at $t_k$, divided by $\beta^k$:
\[\Psi_k = (\overline{s}_k-\overline{s}_{k-1} - \overline{c}_k)+\beta(\overline{s}_{k+1}-\overline{s}_{k} - \overline{c}_{k+1})+ \beta^2(\overline{s}_{k+2}-\overline{s}_{k+1} - \overline{c}_{k+2}) + \ldots\]

Note that $\Psi_k$ is weakly positive, because each Trader can guarantee a zero payoff by repeating the previous announcement and not buying any costly information. Additionally, it is strictly positive if \is expected instant opportunity is strictly positive and it is \is turn to make an announcement. That is, with some probability, some history $H^{k-1}$ occurs and \is instant opportunity is strictly positive.

Consider now $\underset{K \rightarrow \infty}{\lim}\underset{k = 1}{\overset{K}{\sum}} \Psi_k$. On the one hand, this limit must be infinite, because each $\Psi_k$ is non-negative and an infinite number of them are greater than $\nu^*>0$, from Step 3. On the other hand, for any $K$, we have
\begin{align*}
\underset{k = 1}{\overset{K}{\sum}} \Psi_k &= (\overline{s}_{1}-\overline{s}_{0}-\overline{c}_{1}) + \beta(\overline{s}_{2}-\overline{s}_{1}-\overline{c}_{2}) +\beta^2(\overline{s}_{3}-\overline{s}_{2}-\overline{c}_{3})+ \ldots \\
       &+(\overline{s}_{2}-\overline{s}_{1}-\overline{c}_{2})+ \beta(\overline{s}_{3}-\overline{s}_{2}-\overline{c}_{3}) +\beta^2(\overline{s}_{4}-\overline{s}_{3}-\overline{c}_{4})+ \ldots\\
       &+ \;\;\;\;\;\; \vdots\\
       &+(\overline{s}_{K}-\overline{s}_{K-1}-\overline{c}_{K})+ \beta(\overline{s}_{K+1}-\overline{s}_{K}-\overline{c}_{K+1}) +\beta^2(\overline{s}_{K+2}-\overline{s}_{K+1}-\overline{c}_{K+2})+ \ldots\\
       &=(\overline{s}_{K}-\overline{s}_{0}-\underset{k=1}{\overset{K}{\sum}}\overline{c}_{k})+ \beta(\overline{s}_{K+1}-\overline{s}_{1}-\underset{k=2}{\overset{K+1}{\sum}}\overline{c}_{k}) +\beta^2(\overline{s}_{K+2}-\overline{s}_{2}-\underset{k=3}{\overset{K+2}{\sum}}\overline{c}_{k})+ \ldots\\
       &\leq (\overline{s}_{K}-\overline{s}_{0})+ \beta(\overline{s}_{K+1} -\overline{s}_{1}) +\beta^2(\overline{s}_{K+2}-\overline{s}_{2})+ \ldots\\
       &\leq 2M/(1-\beta),
\end{align*}
where costs drop out as they are negative and $M=\underset{y\in [\underline{y},\overline{y}], \omega \in \Omega}{max}|s(y,X(\omega))|$.  But this is impossible, hence $y_k$ must converge in probability  to the true value of $X$.



For the second part, suppose that $X$ is $\kappa$ non-separable. Then, there exist $\mu$ and $v$ at which there is no information acquisition and $X$ is non-separable at $\mu$. Consider the following strategy profile. In every period and after every history, each Trader repeats the previous announcement. The initial announcement by the market maker is $v$.%
\footnote{The initial announcement by the market maker is irrelevant, as long as Trader 1's first announcement is $v$ and then all other traders repeat the previous announcement.}

We will show that this strategy profile is a Nash equilibrium. Note that everyone's expected utility is 0 from following the Nash equilibrium strategy. Suppose that Trader $i$ has a profitable deviation, hence his expected utility, the sum from all the periods where he is making an announcement, is strictly positive. From the definition of $\kappa$ non-separability, the expected utility of Trader $i$ at the first time $t_1$ where he deviates is weakly negative, even if he chooses the optimal experiment and makes the myopic best announcement. Because the sum over all periods is strictly positive, there is time $t_{k'}$ where the expected utility at $t_{k'}$ from deviating is strictly positive. 

We say that Trader $i$ moves the announcement at time $t_k$ if he does not repeat the announcement that Trader $i-1$ made at $t_{k-1}$. Suppose that at time $t_{k'}$ he moves the announcement and let $t_{k}$ be the previous time where he also moved the announcement (hence $t_{k'}$ cannot be the first time he deviates). We will show that \is total expected utility strictly increases if, instead of moving at $t_k$ and at $t_{k'}$, he combines the two announcements and only moves at $t_{k}$. Intuitively, if at $t_k$ he moves the announcement from $v$ to $v'$, and at $t_{k'}$ he moves it from $v'$ to $v''$, he will be strictly better off if he only moves the announcement from $v$ to $v''$ at $t_k$. The reason is that the strictly positive expected utility at $t_{k'}$ is moved earlier at $t_k$, hence it is discounted with $\beta^{t_k}$, which is larger than $\beta^{t_{k'}}$. Recall that all other traders always repeat the previous announcement.

At time $t_k$, suppose that the previous announcement is $y_{t_{k-1}}$, Trader $i$ has initial belief $\mu_k$ and purchases experiment ${\cal R}_{\mu_k}$ (for simplicity, we also denote the resulting random posterior as ${\cal R}_{\mu_k}$).%
\footnote{Note that $\mu_k$ is the belief at $t_k$ after receiving a signal realization at the immediately previous time when he bought an experiment. If he has never bought any experiment before, it is just $\mu$.}
After each signal realization and corresponding posterior $\gamma_k \in Supp({\cal R}_{\mu_k})$, he makes an announcement $y_{\gamma_k}$. His expected utility at $t_k$, given belief $\mu_k$, is 


\begin{equation}
\label{eq t_{k}}
\beta^{t_k} \mu_k\Bigg (\underset{\gamma_k \in Supp({\cal R}_{\mu_{k}})}{\sum} {\cal R}_{\mu_{k}}(\gamma_k) 
E_{\gamma_k} \Big [s(y_{\gamma_k})-s(y_{t_{k-1}})\Big ] -cK({\cal R}_{\mu_k}) \Bigg ),
\end{equation}
where we write $\underset{\omega' \in \Omega}{\sum} \gamma_k(\omega') \Big [s(y_{\gamma_k},X(\omega'))-s(y_{t_{k-1}},X(\omega')) \Big ]$ as $E_{\gamma_k} \Big [s(y_{\gamma_k})-s(y_{t_{k-1}})\Big ]$.

At time $t_{k'}$, given each posterior $\gamma_k \in Supp({\cal R}_{\mu_k})$, he buys another experiment ${\cal R}_{\gamma_k}$ and for each posterior $\gamma_{k'} \in Supp({\cal R}_{\gamma_k})$ he makes an announcement $y_{\gamma_{k'}}$, when the previous announcement is $y_{\gamma_k}$ (since all other traders keep repeating $y_{\gamma_k}$ from $t_{k}$ to $t_{k'}$). His expected utility at $t_{k'}$, given belief $\mu_k$, is

\begin{equation}
\label{eq t_{k'}}
\beta^{t_{k'}}\mu_k \Bigg (\underset{\gamma_k \in Supp({\cal R}_{\mu_{k}})}{\sum} {\cal R}_{\mu_{k}}(\gamma_k)  \Bigg ( \underset{\gamma_{k'} \in Supp({\cal R}_{\gamma_{k}})}{\sum} E_{\gamma_{k'}} \Big [s(y_{\gamma_{k'}})-s(y_{\gamma_k})\Big ] -cK({\cal R}_{\gamma_k}) \Bigg ) \Bigg ).
\end{equation}

If (\ref{eq t_{k'}}) is strictly positive, Trader $i$ could increase his expected utility by bringing forward these announcements, from $t_{k'}$ to $t_k$. The sum of the utilities would be

\begin{equation*}
\beta^{t_k} \mu_k\Bigg (\underset{\gamma_k \in Supp({\cal R}_{\mu_{k}})}{\sum} {\cal R}_{\mu_{k}}(\gamma_k) \Bigg [
E_{\gamma_k} \Big [s(y_{\gamma_k})-s(y_{t_{k-1}})\Big ] + 
\end{equation*}
\begin{equation*} \Bigg ( \underset{\gamma_{k'} \in Supp({\cal R}_{\gamma_{k}})}{\sum} E_{\gamma_{k'}} \Big [s(y_{\gamma_{k'}})-s(y_{\gamma_k})\Big ] -cK({\cal R}_{\gamma_k}) \Bigg ) \Bigg ] -cK({\cal R}_{\mu_k}) \Bigg ) = 
\end{equation*}

\begin{equation}
\label{eq combined}
\beta^{t_k} \mu_k\Bigg (\underset{\gamma_k \in Supp({\cal R}_{\mu_{k}})}{\sum} {\cal R}_{\mu_{k}}(\gamma_k)  \Bigg [ \Bigg ( \underset{\gamma_{k'} \in Supp({\cal R}_{\gamma_{k}})}{\sum} E_{\gamma_{k'}} \Big [s(y_{\gamma_{k'}})-s(y_{t_{k-1}})\Big ] -cK({\cal R}_{\gamma_k}) \Bigg ) \Bigg ] -cK({\cal R}_{\mu_k}) \Bigg ). 
\end{equation}

If (\ref{eq t_{k'}}) is strictly positive, then (\ref{eq combined}) is strictly greater  than (\ref{eq t_{k}}) + (\ref{eq t_{k'}}). Moreover, Trader $i$ does not have to make two announcements at $t_k$. Instead of moving the announcement from $y_{t_{k-1}}$ to $y_{\gamma_k}$ and then to $y_{\gamma_{k'}}$, he can just move the announcement from $y_{t_{k-1}}$ to $y_{\gamma_{k'}}$. From Assumption \ref{ass: costly sequence of signals}, he is able to do this with lower cost. 

It may be that (\ref{eq t_{k'}}), Trader \is expected utility at $t_{k'}$ given belief $\mu_k$, is not positive. However, we have assumed that the ex ante utility at $t_{k'}$ (the expectation over all $\mu_k$ that arise given all signal realizations before $t_k$) is strictly positive. This means that bringing forward all announcements from $t_{k'}$ to $t_{k}$ will strictly increase \is ex ante expected utility. 

Because the sum of expected utilities over all periods is strictly positive, there is finite time $t_{k_T}$ such that the sum of expected utilities over all periods up to $t_{k_T}$ is also strictly positive. Without loss of generality, the expected utility at $t_{k_T}$ is strictly positive, otherwise we can consider an earlier $t_{k_T}$ and the total sum will still be strictly positive.  Consider the following process which starts at $t_{k_T}$ and ends at $t_1$, the time where Trader $i$ deviates initially. If the expected utility at $t_k$, $1 < k \leq T$, is negative, remove it from the sum. Then, the remaining sum increases and therefore remains strictly positive. If the expected utility at $t_k$ is positive, bring forward the announcements to $t_{k-1}$. As described above, the sum of expected utilities increases and therefore it remains strictly positive. As there are finitely many steps, we will reach the first period $t_1$, where Trader $i$ starts his deviation. From this process, the expected utility at $t_1$ is strictly positive. However, this contradicts the fact that at the first time where Trader $i$ deviates, his expected utility is weakly negative,  even if he always makes the myopically best announcement, because the security is $\kappa$ non-separable. Hence, there is no profitable deviation.

\end{proofEnd}

The second bullet point shows that if security $X$ is $\kappa$ non-separable, then the strategy profile where everyone repeats the previous announcement is a Nash equilibrium with no information aggregation. This differs from the no information acquisition setting in \cite{ostrovsky12}, who shows that if $X$ is non-separable then the strategy profile where everyone announces the expected value $v$ is a Perfect Bayesian Equilibrium. To understand the difference between the two settings, note that in both cases the relevant deviation is myopically optimal. In the setting without information acquisition, if Trader $i$ deviates by announcing $v'$, this does not alter his beliefs. The other traders can safely ignore the deviation and keep announcing $v$, because their beliefs are unchanged along this off-equilibrium path. In that case, the constant-announcement profile can be supported as a Perfect Bayesian Equilibrium. 

In the information acquisition setting, by contrast, a deviation in which Trader $i$ acquires a signal and announces $v'$ alters his beliefs. If the other traders keep repeating $v$, then Trader $i$ will make profits in the subsequent periods because now $v'$ is the myopically optimal announcement for him. Hence, the constant-announcement profile is no longer a Nash equilibrium. Instead, we consider the strategy profile where all traders repeat the previous announcement and the proof is more involved. At the first period in which Trader $i$ deviates, his payoff from acquiring information and making the myopically best announcement $v'$ is weakly negative by 
$\kappa$ non-separability and zero in all subsequent periods, as other traders repeat $v'$. Can he profit by buying signals in more than one period? If such a deviation was profitable, he could combine all subsequent signals into one signal in period $t$, and from Assumption \ref{ass: costly sequence of signals} the cost of the combined signal is not higher than the sum of the individual signals. Because of discounting, bringing the gains forward will make him better off. But then, this contradicts the result that a one-period deviation is not profitable. We therefore obtain a Nash equilibrium with no information aggregation. We do not claim here that this strategy profile is also a Perfect Bayesian Equilibrium, because establishing it would require specifying off-equilibrium beliefs after every off-path announcement, which is nontrivial in the infinite-action setting.

Although Theorem \ref{thm:info aggregation} ensures that information aggregates in all Nash equilibria, it does not guarantee that an equilibrium exists. This is true also in \cite{ostrovsky12}, and the reason is that the action spaces are infinite. In Appendix \ref{existence}, we prove the existence of a Perfect Bayesian Equilibrium if the action spaces are finite. We first show that if we discretize the action spaces, so that traders can choose from a finite set of announcements and signals, then a Perfect Bayesian Equilibrium exists whenever the time horizon is finite. We then show that each game $\Gamma$ is continuous at infinity. Finally, we adapt the proofs of \cite{fudenbergLevine83, fudenbergLevine86} to approximate the infinite horizon game with a sequence of finite horizon games and show that the sequence of Perfect Bayesian equilibria in the games with finite horizon converges to a Perfect Bayesian equilibrium in the game with infinite horizon.%
\footnote{\cite{ostrovsky12} describes such an approach but does not provide further details. Existence of equilibrium is also proven by \cite{galanisEtAl24}, for the identical game without information acquisition but with ambiguity aversion.}
We also show that the sequence of announcements converges to a limit which is as close as possible to the correct value of the security. Hence, information aggregation improves as the action space becomes richer.

The second part of our main result says that if security $X$ pays differently across all states, then it is $\kappa$ separable for some marginal cost of information $c$, given {\it any} information structure $\Pi$ and cost function $K$. In other words, if information is sufficiently cheap to acquire, there are easily describable securities where information aggregation is achieved at all Nash equilibria and independently of who participates in the market or what is their private information.  

\begin{theoremEnd}{proposition} \label{prop:different values}
If $X(\omega) \neq X(\omega')$ for all $\omega, \omega' \in \Omega$, then, given any $\Pi$ and $K$, $X$ is $\kappa$ separable for some $c>0$.
\end{theoremEnd}

\begin{proofEnd}
Given $\Pi$ and $K$, we will show that if there is no information acquisition at some $\mu$ and for sufficiently low $c$, then traders cannot agree on the announcement, hence the security is $\kappa$ separable at $\mu$. When the marginal cost of information acquisition $c$ is sufficiently low, and given that $X$ pays differently across all states, the fact that $K$ is sufficiently smooth around zero information (Condition (ii) in Definition \ref{def:cost experiments}) implies that there is no information acquisition only if all posterior beliefs assign a probability close to 1 to some state $\omega' \in \Pi_i(\omega)$, for all $i \in I$ and $\omega \in Supp(\mu)$.  Suppose that there exists $\mu$ with $E_{\mu}[X|\Pi_{i}(\omega)]=v$ for all $i=1,...,n$ and $\omega \in Supp(\mu)$, where $v$ is arbitrarily close to some $m$ with $X(\omega) = m$.
 
 For the security to be $\kappa$ non-separable,  there should also be uncertainty about its value. This means that for some $\omega' \in Supp(\mu)$, $X(\omega') \neq v$.  From the assumption in Section \ref{preliminaries} that $\underset{i \in I}{\bigcap} \Pi_i(\omega) = \{\omega\}$, there exists Trader $i$ such that $m \notin X(\Pi_i(\omega'))$, so that he considers $m$ impossible. Because he assigns a probability close to 1 to some state, from continuity it must be $E_{\mu}[X|\Pi_{i}(\omega')] \neq v$,  a contradiction. As this is true for all $\mu$ that assign probability close to 1 to some state $\omega' \in \Pi_i(\omega)$, we have that the security $X$ is $\kappa$ separable for some $c$. 
    
\end{proofEnd}

Proposition \ref{prop:different values} echoes the following result within the context of perfectly competitive markets, as pointed out in \cite{laffontMaskin90}. As long as there are ``enough prices'', so that trade can be made contingent on sufficiently many events, then the competitive equilibrium is generically separating and the function relating information to prices is invertible (\cite{grossman76}, \cite{radner79}, \cite{allen81}). We could interpret Proposition \ref{prop:different values} as a generalization of this result, because our setting applies to all Nash equilibria, not just competitive equilibria. 

It would be natural to expect that, for a sufficiently low marginal cost $c$, any security eventually becomes $\kappa$ separable. But this is not true. Consider the following example, which appears in similar form in \cite{geaPol82} and in Example 1 of \cite{ostrovsky12}.

\begin{Example} \label{example geanakoplos polemarchakis}
The state space is $\Omega = \{\omega_1, \omega_2, \omega_3, \omega_4\}$, the security is $X= (0,1,0,1)$, and there are two traders with common prior $\mu = (1/4,1/4,1/4,1/4)$. Trader 1's partition is $\{\{\omega_1, \omega_2\}, \{\omega_3, \omega_4\}\}$ and Trader 2's is $\{\{\omega_1, \omega_4\}, \{\omega_2, \omega_3\}\}$.
\end{Example}


With a uniform prior, all traders agree that the expected value of $X$ is 0.5, at all states, hence the security is non-separable. Consider now  prior $\mu = (\frac{1-2m}{2},m, \frac{1-2m}{2},m)$, where $0 < m <0.5$. Given any $m$, each Trader's expected value of $X$ at all states is $2m$, so for $m$ very close to $0.5$, the consensus announcement is very close to $1$. Given $K$ and for any $c>0$, we can always specify $m$ that is sufficiently close to $0.5$, so that no Trader is willing to buy a signal ${\cal R}$, because this will move their expected value of $X$ from $1-2\epsilon$, which is the previous announcement, to $1-\epsilon$ when they receive the signal realisation, for small enough $\epsilon>0$. Their maximum benefit, when 1 is the true state, is $s(1-\epsilon, 1) - s(1-2\epsilon,1)$. Because the scoring rule is proper, this benefit is strictly positive but decreasing, as $\epsilon$ decreases.  However, as the signal becomes ever more precise, its cost $cK({\cal R})$ is weakly increasing. This means that for high enough $m <0.5$, the expected benefit will be lower than the cost. Because there is no information acquisition and all traders agree on the expected value, security $X$ is $\kappa = (K, c)$ non-separable for all $c>0$.%
\footnote{It is interesting to note that the incompatibility of the  $(0,1,0,1)$ vector with information aggregation is met in other settings as well. In the model of \cite{degroot74}, agents update their beliefs naively, by looking at their immediate neighbours, according to a fixed network. If the network is represented by a periodic matrix, such as $A=\begin{pmatrix}
  0 & 1\\ 
  1 & 0
\end{pmatrix}$, then beliefs do not converge (\cite{golubJackson10}).}

The following Proposition generalizes this example, showing that the only securities that never become $\kappa$ separable pay $(a, d, b, d)$, in four states,  where either $a,b < d$ or $a,b > d$. 

\begin{theoremEnd}{proposition}\label{prop:kappa non-separable characterization}
If $\Omega$ has at least four states, the following are equivalent:
\begin{itemize}
    \item $X$ is $\kappa$ non-separable given some $\Pi$ and for all $\kappa \in {\cal K}$,
    \item $X$ pays $(a, d,b, d)$ in four states, where either $a,b < d$ or $a,b > d$.
\end{itemize}
\end{theoremEnd}

\begin{proofEnd}
Suppose $\Omega$ has at least four states and let $E = \{\omega_1, \omega_2, \omega_3, \omega_4\}$, where $X$ pays $(a, d, b, d)$, with either $a,b < d$ or $a,b > d$.  What $X$ pays outside $E$ is irrelevant, because we will only consider priors that have full support on $E$, in order to show that $X$ is $\kappa$ non-separable given some $\Pi$ and for all $\kappa$. Consider two traders with the following information structure on $E$. Trader 1's partition is $\{\{\omega_1, \omega_2\}, \{\omega_3, \omega_4\}\}$ and Trader 2's is $\{\{\omega_1, \omega_4\}, \{\omega_2, \omega_3\}\}$. We need to show that for any $\kappa = (K,c)$, $X$ is $\kappa$ non-separable. It is enough to show that, as the marginal cost $c$ converges to 0, there is always a common prior $\mu$ on $E$ such that no Trader acquires any information and $X$ is non-separable at $\mu$. 

Consider  prior $\mu = (p,m, 1-p-2m,m)$, where $0 < m <0.5$ and $0< p <1-2m$. We specify $p$ such that, given any $m$, each Trader's expected value of $X$ is the same at all states. In the following equation, the left hand-side computes Trader 1's expected value of $X$ at $\{\omega_1, \omega_2\}$ and Trader 2's at $\{\omega_1, \omega_4\}$, whereas the right hand-side computes Trader 1's expected value at $\{\omega_3, \omega_4\}$ Trader 2's at $\{\omega_2, \omega_3\}$:
\begin{equation} \label{eq for p}
   a \frac{p}{p+m} + d \frac{m}{p+m} = b \frac{1-p-2m}{1-p-m} + d \frac{m}{1-p-m}. 
\end{equation}

There is always a solution for some $p \in (0, 1-2m)$ to this equation. To see this, group the four terms of (\ref{eq for p}) on one side. If $p$ is very close to zero, we have
\[ d - b \frac{1-2m}{1-m} - d \frac{m}{1-m} = \frac{d-2dm-b(1-2m)}{1-m} = \]
\[\frac{d(1-2m)-b(1-2m)}{1-m} = \frac{(d-b)(1-2m)}{1-m} \]

If $p$ is very close to $1-2m$ we have
\[a \frac{1-2m}{1-m} + d \frac{m}{1-m} - d = \frac{a(1-2m) - d(1-2m)}{1-m} = \frac{(a-d)(1-2m)}{1-m}\]
In both cases, $a,b < d$ or $a,b > d$, one expression is strictly positive and the other is strictly negative, because $1-m$ and $1-2m$ are always positive. By continuity, there exists $p$ at which the expression is zero, hence there is a solution.

Let $u$ be the expected value of $X$ at all states and for all traders. Fix any marginal cost $c$, which may be very close to 0. We can then choose $m$ very close to 0.5 (and corresponding $\mu$) such that $u$ is very close to $d$, say $d-\epsilon$, assuming, without loss of generality, that $a,b < d$, otherwise we have $d+\epsilon$ and the rest of the proof proceeds accordingly.  By repeating the consensus announcement of $u = d-\epsilon$ and not acquiring any signals, both traders get 0 utility. Suppose that one Trader wants to buy a signal ${\cal R}$ with random posterior $Q$, which, in the case of $X(\omega) =d$, will move his posterior expected value of $X$ from $d-\epsilon$ to $d-\epsilon+\nu$, for some $0<\nu<\epsilon$. The upper bound of his utility, net of the cost of the signal, is realised at $X(\omega) = d$ and it is $s(d-\epsilon+\nu,d) - s(d-\epsilon,d)$. This is positive because $d-\epsilon+\nu$ is closer to $d$ than the previous announcement, $d-\epsilon$, and the proper scoring rule is order-sensitive (see footnote \ref{fn: ostrovsky p2618}). However, by decreasing $\epsilon$ appropriately it can be made as small as needed, and therefore smaller than $cK({\cal R})$ for all ${\cal R}$.  As the benefit becomes smaller and $\epsilon$ decreases, a signal must be more precise and $\nu$ will decrease so that $d-\epsilon + \nu <d$, hence the cost of the signal will increase weakly, from property (iii) of Definition \ref{def:cost experiments}. Because the benefit converges to zero, at some threshold the cost will be higher than the benefit and no signal ${\cal R}$ will be acquired.
Hence, no Trader will buy any signal and $X$ is $\kappa$ non-separable.

For the converse, suppose $X$ is $\kappa$ non-separable given some $\Pi$ and for all $\kappa$. The contrapositive of Proposition \ref{prop:different values} 
implies that there exist  two states $\omega_a \neq \omega_b$ such that $X(\omega_a)=X(\omega_b) = m$. Let $M$ be the set of states $\omega$ with $X(\omega) = m$. Given that $\Omega$ has at least four states, we have the following cases.


Case 0. $M = \Omega$. Security $X$ is constant and therefore (trivially) always separable, for all information structures $\Pi$. This is a contradiction because we have assumed $X$ is $\kappa$ non-separable given some $\Pi$ and for all $\kappa$.

Case 1. There is either a unique $\omega \notin M$, or there are exactly two states $\omega,\omega' \notin M$, such that $X(\omega) > m > X(\omega')$. From Proposition \ref{prop:always separable}, $X$ is always separable, for all information structures $\Pi$. As with Case 0, this is a contradiction.
The only other remaining case is that there exist $\omega, \omega' \notin M$ with either $X(\omega), X(\omega') < m$, or $X(\omega), X(\omega') > m$, which concludes the proof.



\end{proofEnd}

To interpret this result, note that state space $\Omega$ is fixed independently of the securities and it describes all the higher order beliefs of traders about what others know and believe. The Proposition says that information aggregation may fail if there are two distinct states which specify the exact same payoff, yet they differ in terms of the higher order beliefs of traders and the only thing that matters to them are the payoffs. 

\section{Classification of Securities}
\label{classification of securities}

Using the results from Section \ref{main result}, we can provide a complete classification of securities in terms of how well they aggregate information, which surprisingly depends only on their payoff structure. There are three types of securities. First, the {\it always separable} securities, which we characterise in this section, are separable for {\it all} information structures and therefore aggregate information even if information acquisition was unavailable. Second, the $\kappa$ non-separable for all $\kappa$ and {\it some} information structure, such as the $(1,0,1,0)$ security, which may fail information aggregation even if $c>0$ is very small.  Finally, the securities which are $\kappa$ separable for some $\kappa$ and {\it all} information structures, pay differently across all states and  aggregate information if $c$ is low enough. This classification does not depend on who is trading or what is their information structure. Moreover, it is very simple to classify each security, which is not true when determining whether a security is separable or not, given an information structure.

The always separable securities aggregate information even if costly signals are not available. Unfortunately, this class is very small and uninformative, as shown by the following Proposition. It consists of just three types. The first is the constant, paying the same at all states. The second is  the Arrow-Debreu (A-D), which pays $a$ at some state $\omega$ and $b$ at all other states. The third pays $a$ at some state $\omega$, $d$ at $\omega'$, and $b$ at all other states, where $a < b < d$. 



\begin{theoremEnd}{proposition}\label{prop:always separable}
 The only non-constant securities that are separable for all information structures in ${\cal P}$ are the A-D and the security that is of the following form. There are values $a<b<d$ such that $X(\omega_a)=a$ and $X(\omega_d) = d$ for two states $\omega_a, \omega_d$, and $X(\omega) = b$ for all $\omega \neq \omega_a, \omega_d$. 
\end{theoremEnd}

\begin{proofEnd} We start by restating a useful characterization of separable securities by \cite{ostrovsky12}. It specifies that $X$ is separable if and only if, for any possible announcement $v$,  we can find numbers $\lambda_i(\Pi_i(\omega))$, for each $i$ and $\omega$, such that the sum over all traders has the same sign as the difference of $X(\omega) -v$. Intuitively, for any $v$ and at each $\omega$, all traders ``vote'' and the sign of the sum of the votes has to agree with the sign of the difference between the value of the security and $v$.

\begin{proposition}[\cite{ostrovsky12}] \label{ostrovsky char thm}
Security $X$ is separable under partition structure $\Pi$ if and only if, for every $v \in {\mathbb R}$, there exist functions $\lambda_i: \Pi_i \rightarrow {\mathbb R}$ for $i = 1, \ldots, n$ such that, for every state $\omega$ with $X(\omega) \neq v$, 
\[(X(\omega) - v)\underset{i \in I}{\sum} \lambda_i(\Pi_i(\omega)) > 0.\]
\end{proposition}

If $\Omega$ has up to three states, then all securities are of the two types that we have described or the uninteresting case of a constant security. Hence, without loss of generality, we fix a state space $\Omega$ with at least four states and a security $X$. If $X$ is constant, it is trivially separable. \cite{ostrovsky12} shows that an A-D security is always separable. 

We now show that $X$ is always separable if it is of the following form. Suppose there are $a<b<c$ such that $X(\omega_a)=a$ and $X(\omega_c) = c$ for two states $\omega_a, \omega_c$, whereas $X(\omega) = b$ for all $\omega \neq \omega_a, \omega_c$.%
\footnote{Our proof for this type of security also applies to an A-D security. \cite{ostrovsky12} used Corollary 1 to show that an A-D security is always separable, however, we cannot use it for this type of security.}
Using Proposition \ref{ostrovsky char thm}, we need to show that for every $v \in {\mathbb R}$, there exist functions $\lambda_i: \Pi_i \rightarrow {\mathbb R}$ for $i = 1, \ldots, n$ such that, for every state $\omega$ with $X(\omega) \neq v$, 
\begin{equation} \label{char condition}
(X(\omega) - v)\underset{i \in I}{\sum} \lambda_i(\Pi_i(\omega)) > 0.
\end{equation}

If $v\leq a$, then condition (\ref{char condition}) is satisfied by setting $\lambda_i(\Pi_i(\omega)) = 1$ for all $i \in I$ and $\omega \in \Omega$. Similarly, if $v \geq c$, we set $\lambda_i(\Pi_i(\omega)) = -1$ for all $i \in I$ and $\omega \in \Omega$. Suppose that $a<v \leq b<c$. For all $i \in I$, set $\lambda_i(\Pi_i(\omega_a)) = -1$ and (\ref{char condition}) is satisfied for $\omega_a$. For all $\omega$ with $\Pi_i(\omega) \neq \Pi_i(\omega_a)$, set $\lambda_i(\Pi_i(\omega)) = k$, where $k = |I|$ is the number of agents. Because of our assumption that the join of all partitions consists of singleton sets, we have that for each $\omega \neq \omega_a$, there exists $i$ such that $\Pi_i(\omega) \neq \Pi_i(\omega_a)$. This implies that if $X(\omega) - v > 0$, we also have $\underset{i \in I}{\sum} \lambda_i(\Pi_i(\omega)) \geq k-(k-1)>0$  and (\ref{char condition}) is satisfied for $\omega$. Using a symmetric argument, we can show that (\ref{char condition}) is satisfied for $a< b\leq v < c$, by setting $\lambda_i(\Pi_i(\omega_c)) = 1$ and $\lambda_i(\Pi_i(\omega)) = - k$ for $\omega$ with $\Pi_i(\omega) \neq \Pi_i(\omega_c)$, for all $i \in I$. By applying Proposition \ref{ostrovsky char thm}, security $X$ is always separable.

Suppose that $X$ is not of the three aforementioned types. Then, we can find four distinct states where $X$ assigns values $a \leq b < c \leq d$. 
For simplicity, we refer to the state with value $a$ as state $a$ and similarly for $b,c$, and $d$. 

We will show that $X$ is non-separable for an information structure in ${\cal P}$ with two agents. The partition of Trader 1 is $\{\{a,d\},\{b,c\}\}$ for these four states, whereas for any other state, we have $\Pi_1(\omega) = \{\omega\}$. For Trader 2 it is $\{\{a,c\},\{b,d\}\}$ and for any other state we have $\Pi_2(\omega) = \{\omega\}$. Hence, the information structure is in ${\cal P}$.

To show that $X$ is non-separable, it is enough to find a prior $p$ with support on $\{a,b,c,d\}$ such that, for some $v$,

\begin{itemize}
\item[$(i)$] $X(\omega) \neq v$ for some $\omega \in Supp(p)$,
\item[$(ii)$] $E_{p}[X|\Pi_{i}(\omega)]=v$ for $i=1,2$ and $\omega \in Supp(p)$. 
\end{itemize}

Let $p_1$ be 1's probability of state $a$ conditional on $\{a,d\}$, whereas $q_1$ is 1's probability of state $b$ conditional on $\{b,c\}$. Let $p_2$ be 2's probability of state $a$ conditional on $\{a,c\}$, whereas $q_2$ is 2's probability of state $b$ conditional on $\{b,d\}$. Condition (ii) then translates to the following equations 
 \begin{equation} \label{cond ii}
\begin{cases}
ap_1 + d(1 - p_1) =  bq_1 + c(1 - q_1) \\
ap_1 + d(1 - p_1) =  ap_2 + c(1 - p_2) \\
ap_1 + d(1 - p_1) =  bq_2 + d(1 - q_2)\\
bq_1 + c(1 - q_1) =  ap_2 + c(1 - p_2) \\
ap_1 + d(1 - p_1) = bq_2 + d(1 - q_2) \\
ap_2 + c(1 - p_2) =  bq_2 + d(1 - q_2) 
\end{cases}
\end{equation}
 
 The posteriors of the two agents can be derived by a common prior $p$ if the following conditions hold:
 \begin{equation} \label{cond cp}
\begin{cases}
xp_1 = yp_2\\
(1-x)q_1 = (1 - y)q_2\\
(1 - x)(1 - q_1) = y(1 - p_2)\\
 x(1 - p_1) = (1-y)(1 - q_2)
\end{cases}
\end{equation}
 
where $x$ is the prior probability of $(a,d)$ and $y$ is the prior probability of $(a,c)$.
 
When $a \leq b<c \leq d$, the system (\ref{cond ii} - \ref{cond cp})  has the following solution:
\begin{equation*} 
\begin{split} 
q_1 = \frac{a - c}{a + b - c - d} ,\\
p_1 = \frac{b - d}{a + b - c - d} ,\\
p_2 = \frac{b - c}{a + b - c - d} ,\\
q_2 = \frac{a - d}{a + b - c - d} ,\\
x = p_2,\\
y = p_1. 
\end{split}
\end{equation*}
These posteriors uniquely define the respective prior probabilities.
\end{proofEnd}

All three types of securities are not very informative. Even when there is information aggregation and the price of $X$ always converges to the true value $X(\omega)$ at state $\omega$, the A-D only reveals whether $\omega$ has occurred or not, and the second type only reveals that either $\omega$, $\omega'$, or neither, have occurred. In other words, the price of $X$ does not reveal information about most events in $\Omega$. In contrast, the most informative security $X'$ pays differently across all states. If there is information aggregation, then $X'$ reveals whether any event in $\Omega$ has occurred or not. However, because $X'$ is not always separable, we know that for some information structure $\Pi$, information does not aggregate.




We now formally state our classification. Let $X$ be a security defined on $\Omega$ that has at least four states. If there is $m$ such that $X(\omega') = X(\omega'') = m$ for at least two states $\omega', \omega'' \in \Omega$, let $M$ be the set of states $\omega$ with $X(\omega) = m$. Otherwise, set $M = \emptyset$.

{\bf Case 0.} $M = \Omega$. Security $X$ is constant and therefore (trivially) always separable, for all information structures $\Pi$.

{\bf Case 1.} There is either a unique $\omega \notin M$, so that the security is A-D, or there are exactly two states $\omega,\omega' \notin M$, such that $X(\omega) > m > X(\omega')$. From Proposition \ref{prop:always separable}, $X$ is always separable, for all information structures $\Pi$.

{\bf Case 2.} There exist $\omega, \omega' \notin M$ with either $X(\omega), X(\omega') < m$, or $X(\omega), X(\omega') > m$. From Proposition \ref{prop:kappa non-separable characterization}, $X$ is $\kappa$ non-separable for all $\kappa$, for some information structure $\Pi$. 

{\bf Case 3.} $M = \emptyset$, hence $X(\omega) \neq X(\omega')$, for all $\omega, \omega' \in \Omega$. From Proposition \ref{prop:different values}, $X$ is $\kappa$ separable for some $\kappa$, for all information structures $\Pi$.


We conclude the section with a few remarks on the properties of $\kappa$ separable securities. First,  separability implies $\kappa$ separability, for all $\kappa$, because there does not exist a prior at which the security is non-separable. Second, non-separability implies $\kappa$ non-separability for some $\kappa$, because for a high enough marginal cost $c$, no Trader will acquire any information.

\begin{Remark}
If security $X$ is separable given an information structure $\Pi \in {\cal P}$, then it is also $\kappa$ separable for any cost structure $\kappa \in {\cal K}$.
\end{Remark}

 
\begin{Remark}
 If security $X$ is non-separable given an information structure $\Pi \in {\cal P}$, then for any cost function $K$ it is $\kappa = (K,c)$ non-separable for some marginal cost of information $c$.
\end{Remark}

Note that if $X$ is $\kappa = (K,c)$ non-separable, then there is a $\mu$ for which there is no information acquisition and $X$ is non-separable at $\mu$. If we increase the marginal cost to $c'>c$, then there is still no information acquisition for $\mu$, hence $X$ is  $\kappa' = (K,c')$ non-separable. Conversely, if $X$ is $\kappa = (K,c)$  separable, it is separable for all $\mu$ for which there is no information acquisition. If we decrease the marginal cost to $c' < c$, then the set of priors for which there is no information acquisition will shrink and therefore $X$ will be $\kappa' = (K,c')$ separable for all $c' < c$. We, therefore, have the following remark.

\begin{Remark} \label{remark: lower cost kappa separable}
If $X$ is $\kappa = (K,c)$ non-separable (separable), then it is $\kappa' = (K,c')$ non-separable (separable) for all $c' > c$ ($c' < c$). 
\end{Remark}

\section{The Value of the Market}
\label{sec: value of the market}


If the cost of information drops significantly, do we even need markets to aggregate information through prices? If each Trader can independently afford all relevant signals, then the average of their predictions would be very close to the outcome of the market, rendering it obsolete. In this section, we argue with an example that this intuition is not correct, because markets can aggregate information long before it becomes economically viable for each Trader to acquire the required information on their own. Hence, markets become even more important in an environment with information acquisition. 

To illustrate this point, we compare the market's prediction accuracy to that of a poll where traders simultaneously make a single announcement.  The poll's prediction is then calculated as the average of their individual predictions.%
\footnote{Note that there are many ways of improving the accuracy of a poll by aggregating announcements differently \citep{baron2014two}. In our framework, markets will always be more accurate than polls because more information is disseminated through multiple rounds of announcements, and the value of information is positive. Several papers have examined the two settings in experiments and real-life settings, and the results are mixed.  \cite{snowberg2013prediction} argue that prediction markets are better. \cite{bergEtAl08} show that the Iowa Electronic Markets were more accurate than 964 polls in predicting the outcomes of five presidential elections between 1988 and 2004. \cite{CowgillZitzewitz15} show that internal prediction markets in Google and Ford were more accurate than the predictions of professional forecasters. On the other hand,  \cite{atanasovEtAl16, dana2019markets} argue that while prediction markets are more accurate than the simple mean of forecasts from polls, the latter outperform prediction markets when forecasts are aggregated with transformation algorithms or made in teams. \cite{camererEtAl16} show that markets are equally accurate with a survey in predicting the replicability of economic experiments.}
%

Fix a cost function $K$ and a prior $\mu$ for which the security is non-separable and then gradually decrease $c$, the marginal cost of information. For high enough $c$, the market and the poll are equally accurate. At a threshold $\bar{c}$, the security becomes $\kappa$ separable and the accuracy of the market increases discontinuously to 1, whereas the accuracy of the poll increases gradually as $c$ decreases further.

Recall Example \ref{main example} with state space $\Omega = \{\omega_1, \omega_2, \omega_3, \omega_4\}$, security $X= (0,1,2,3)$, and common prior $\mu_0=(\frac{1}{8},\frac{3}{8},\frac{3}{8},\frac{1}{8})$. Trader 1's partition is $\{\{\omega_1, \omega_3\}, \{\omega_2, \omega_4\}\}$ and Trader 2's is $\{\{\omega_1, \omega_4\}, \{\omega_2, \omega_3\}\}$. Security $X$ is non-separable at $\mu_0$ because the expected value of $X$ for both traders is $v=\frac{3}{2}$ at all states, yet there is uncertainty about the value of the security. Traders can acquire signals where the cost is proportional to the expected reduction in entropy relative to the fixed posterior $\mu=0.5$ for every $\omega$  \citep{shannon1948mathematical}. 
We assume that traders are myopic and proper scoring rules are used in both settings, hence each announcement is the expected value of the security given the acquired information.

Figure~\ref{Fig: MarketV} shows the expected accuracy of markets and polls for prior $\mu$. When the marginal cost of information is high (c>4.5), no Trader acquires any information and the accuracy of the market is equal to that of the poll. As the marginal cost decreases below $4.5$, Trader 2 starts acquiring information if the state is either 1 or 4. This implies that his announcement differs across partitions,  revealing whether event $\{\omega_1, \omega_4\}$ or $\{\omega_2, \omega_3\}$ is true. Trader 1 combines this with his private information and learns which state is true, thus announcing $v=X(\omega)$. Therefore, a small change in the marginal cost of information allows the market to  aggregate information, with a prediction accuracy of 1. In contrast, the poll's accuracy improves gradually as information gets cheaper. This means that, as the cost of information decreases, the prediction accuracy of the market suddenly jumps to 1, whereas the accuracy of the poll gradually increases. 
 
\begin{figure}[h!] 
	\centering
		\includegraphics[width=1\textwidth]{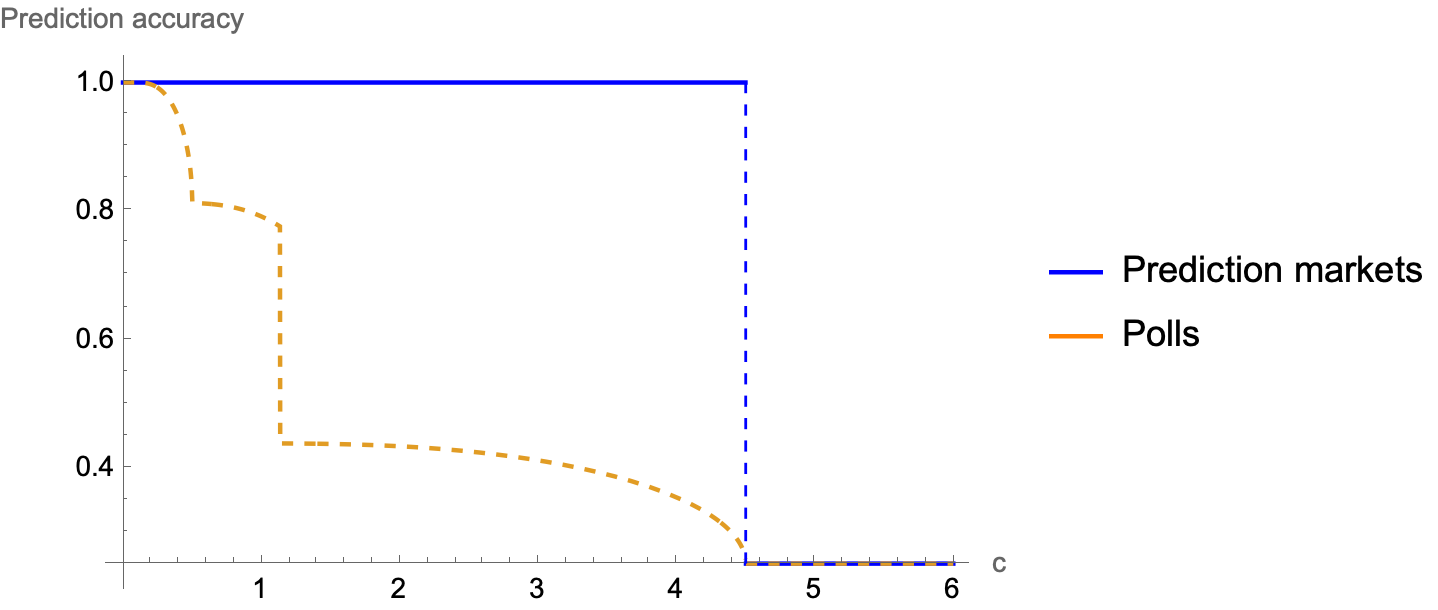}
\caption{Prediction accuracy ($A(\mu) = E_\mu A(\omega,\mu)$) for markets and polls.}\label{Fig: MarketV}
\end{figure}

%


We now provide the formal definitions of a poll and of prediction accuracy. We then state Proposition \ref{prop: char k separable and poll} which characterizes $\kappa$ separable securities in terms of the market being strictly more accurate than a poll for all priors. Recall that the full state $\phi = (\omega, \iota_1, \iota_2, \ldots)$  describes the initial uncertainty $\omega \in \Omega$ and the randomizations of the players, as well as the signal realisations. 

\begin{Definition}
In a poll there is only one round of announcements and traders are myopic. For each $\phi \in \Phi$ and common prior $\mu$, the prediction of the poll is the average of the myopic predictions, $y_i$, where each Trader $i$ optimally obtains a signal ${\cal R}_i$ and then they all announce simultaneously:
$$p^p(\phi,\mu) = \frac{\underset{i=1}{\overset{n}{\sum}} y_i}{n}.$$
The accuracy of the poll at $\phi$ and $\mu$ is $A^p(\phi,\mu) = 1-|p^p(\phi,\mu)- X(\phi)|$.
\end{Definition}



The prediction of the market depends on the equilibrium that is played. Let ${\cal L}(\mu)$ be the set of Nash equilibria given $\mu$. We know that in every Nash equilibrium $l \in {\cal L}$ and state $\phi$, the announcements converge to some price $p^m(\phi,\mu, l)$, which we denote as the prediction of the market. We define the accuracy of the market given $\phi$ and $\mu$ as $A^m(\phi,\mu) = \underset{l \in {\cal L}(\mu)}{\inf} \{1-|p^m(\phi,\mu,l)- X(\phi)|\}$, the worst accuracy among all equilibria. The highest possible accuracy is 1, when $p^m(\phi,\mu, l)= X(\phi)$ at all Nash equilibria. The expected accuracy of the market given $\mu$ is $A^m(\mu) = E_\mu A^m(\phi,\mu)$, whereas for the poll it is  $A^p(\mu) = E_\mu A^p(\phi,\mu)$.





For Proposition \ref{prop: char k separable and poll}, we assume that if the support of ${\cal R}(\cdot |\omega)$ is not the same for all $\omega \in \Omega$, then the cost of ${\cal R}$ is infinite. This implies that it is impossible for a Trader to exclude a state $\omega$ from being possible, just by observing a signal realization. This assumption makes the analysis non-trivial; otherwise, traders would be able to exclude states from being possible by acquiring signals with sufficiently low $c$, so obtaining information from trading behavior would be irrelevant.


\begin{Assumption} \label{ass: infinite cost} 
If $Supp({\cal R}(\cdot |\omega)) \neq Supp({\cal R}(\cdot |\omega'))$ for some $\omega, \omega' \in \Omega$, then $K({\cal R}) = +\infty$.
\end{Assumption}

Recall that a non-degenerate prior $\mu$ given $X$ does not assign probability 1 to a unique value of security $X$. If the security is $\kappa = (K,c)$ separable, then information gets aggregated for all  marginal costs $c' \leq c$, which is not true for polls. We, therefore, have the following remark for each $\Pi \in {\cal P}$ and $\kappa = (K,c)$.

\begin{Remark}\label{remark:marketv}
Under Assumption \ref{ass: infinite cost}, if security $X$ is $\kappa$ separable, then for any non-degenerate $\mu$ given $X$ and for all $0<c'\leq c$, information gets aggregated by the market with cost $\kappa' = (K, c')$ at all Nash equilibria, so that $A^m(\mu) = 1$, but it is not aggregated by the poll,  so that $A^p(\mu) < 1$.
\end{Remark}

To see this, note first that $X$ is $\kappa'$ separable, as shown in Remark \ref{remark: lower cost kappa separable}, hence the first part follows directly from  Theorem~\ref{thm:info aggregation}, which shows that information gets aggregated in all Nash equilibria.  Assumption \ref{ass: infinite cost}
implies that individual traders cannot buy a signal that reveals the true state with probability 1 in the first round. Because $\mu$ is non-degenerate given $X$, we have that $p^p \neq X(\phi)$. 

We can further generalise this result, showing that $\kappa$ separability is equivalent to the market being strictly more accurate than the poll, for all non-degenerate priors given $X$.

\begin{theoremEnd}{proposition}\label{prop: char k separable and poll}
Fix cost structure $\kappa$. Under Assumptions \ref{ass: costly sequence of signals} and \ref{ass: infinite cost}, security $X$ is $\kappa$ separable given $\Pi \in {\cal P}$ if and only if $A^m(\mu,l) > A^p(\mu)$ for all non-degenerate priors $\mu$ given $X$.
\end{theoremEnd}

\begin{proofEnd}
If $X$ is $\kappa$ separable, then from Theorem~\ref{thm:info aggregation} we have that information aggregates at all states and all Nash equilibria, therefore  $A^m(\mu)=1$. Assumption \ref{ass: infinite cost} implies that no Trader would acquire full information in one round, hence $p^p(\phi, \mu) \neq X(\phi)$ for all $\phi \in \Phi$ and $A^p(\mu)<1$.


For the converse, suppose that $X$ is $\kappa$ non-separable. Then, we can find non-degenerate $\mu$ given $X$ for which there is no information acquisition and $X$ is non-separable at $\mu$. Using the proof of the second part of Theorem \ref{thm:info aggregation}, which requires Assumption \ref{ass: costly sequence of signals}, there is a Nash equilibrium where everyone agrees on the announcement at all states in the support of $\mu$ and the game ends in the first round. The announcement is the same for everyone, so the poll gives the same prediction as the market, contradicting that $A^m(\mu) > A^p(\mu)$.
\end{proofEnd}

To interpret this result, let $\Delta_X$ be the set of all non-degenerate priors given $X$. Suppose we define the value of the market (with security $X$) to be $V(X)= \underset{\mu \in \Delta_X}{\min} [A^m(\mu) - A^p(\mu)]$, the minimum improvement in accuracy by switching from the poll to the market. Then, Proposition \ref{prop: char k separable and poll} implies that $X$ is $\kappa$ separable if and only if $V(X)>0$.

\section{Concluding Remarks}
\label{conclusion}

The paper provides a thorough examination of the interplay between information aggregation and information acquisition. We show that $\kappa$ separability characterizes information aggregation when the cost of information acquisition is $\kappa$ and we group securities into three classes. First, the `always separable' securities aggregate information irrespective of who trades or what is their information structure. Second, the securities which are $\kappa$ separable for some $\kappa$ and all information structures, aggregate information if the cost is sufficiently low. 
Finally, there is a small class of securities such that for any $\kappa$, each is $\kappa$ non-separable for some information structure. This means that even if the cost is very close to zero, information may not aggregate. Surprisingly, these three classes are easily distinguishable just by looking at the payoff structure of each security.

An interesting question is whether information aggregates if the security is non-separable but the prior is generic. If a security is non-separable, then there is a common prior where information does not aggregate. However, if we were able to slightly perturb this prior then information could aggregate, because the expected value of the security would not be the same across all agents and all partition cells, hence an `announcement' would reveal public information.   
\cite{ostrovsky12} shows that information aggregates in all pure-strategy equilibria with a generic prior, even for non-separable securities. However, in his setting, it is an open question what happens with mixed-strategy equilibria. Intuitively, even if we perturb the initial prior, and given that agents have infinite action spaces, it is not clear whether they will converge to some belief at which all announcements are equal yet there is some uncertainty about the security. It is therefore `more potent' to be able to perturb the belief at the interim stage, when all agents are stuck at making the same announcement, rather than at the beginning of the game.

The current paper provides a novel perspective on this issue of genericity, by endogenizing the perturbation of beliefs {\it at the interim stage}. If information is not aggregated at time $t$ and all traders agree on the price, a Trader could buy an additional signal, if the cost is not too high, and profit from changing his belief. By examining what happens when the cost is very low, we effectively allow for arbitrarily small perturbations of beliefs to be feasible for traders. A ``generic'' security pays differently across states and Proposition \ref{prop:different values} shows that if the cost is sufficiently low, such a security is $\kappa$ separable, for all information structures. Therefore, using Theorem \ref{thm:info aggregation} we can say that generically (for almost all securities) information gets aggregated, in all mixed-strategy equilibria, thus providing an answer to the open question of \cite{ostrovsky12}.


\newpage
\appendix

\section{Existence of a Perfect Bayesian Equilibrium}
\label{existence}


Perfect Bayesian equilibria do not always exist in games with an infinite horizon and infinite actions. In this section, we discretize the action space and show that, because the game is continuous at infinity,  a Perfect Bayesian equilibrium always exists. Continuity at infinity is achieved by shortening the time period, $t_k$, as $k \rightarrow \infty$, so that the discount factor decreases. This is similar to the approach of \cite{ostrovsky12} and \cite{galanisEtAl24}.

We first establish that a Perfect Bayesian equilibrium always exists in games with finitely many actions and periods. Then, by adapting the proofs of \cite{fudenbergLevine83, fudenbergLevine86}, we approximate the game of infinitely many periods with a sequence of  games with finitely many periods and show that the sequence of Perfect Bayesian equilibria in the finite games converges to a Perfect Bayesian equilibrium in the infinite game. Finally, we show that the announcements converge to a limit which is as close as possible to the true value of $X$, given that only finitely many announcements are available.

\subsection{Games with Finitely Many Actions}


Let $ {\cal Y} \sq [\underline{y}, \overline{y}]$ be a set of finitely many announcements, evenly distributed on $[\underline{y}, \overline{y}]$, and ${\cal E}$ a set of finitely many signals. Each signal ${\cal R} \in {\cal E}$ has finitely many realizations $T_{\cal R}$. We denote a game with $k < \infty$ periods and finitely many actions as $\Gamma^k$. We denote a game with infinitely many periods and finitely many actions as $\Gamma^\infty$. Recall that a game with infinitely many actions and periods is denoted by $\Gamma$. A game with finitely many actions ($\Gamma^k$ or $\Gamma^\infty$) is a tuple  $(\Omega, \Pi, X, \mu, {\cal E}, c, K, y_0, {\cal Y}, s, \beta)$ where $\Omega$ is the state space, $\Pi$ is the collection of partitions, $X$ is the traded security, $\mu$ is the common prior, ${\cal E}$ is the finite set of signals, $c$ is the marginal cost and $K$ is the cost function of signals,  $y_{0}$ is the market maker's initial announcement at time $t_0$, ${\cal Y}$ is the finite set of possible announcements, $s$ is a strictly proper scoring rule, and $\beta$ is the common discount rate.

Recall that other traders are not informed about which signal is chosen by Trader $i$, or of the specific signal realization. Therefore, an information set ${\cal I}_k$ for the Trader who announces at $t_k$ only distinguishes between decision nodes that specify a different history of announcements, or a different history of his own signal realizations.  Let ${\mathscr I}$ denote the collection of all information sets. 

An assessment ${\cal A} = \{\sigma, {\mathscr B} \}$ is a strategy profile $\sigma \in \Sigma$ and a system of beliefs ${\mathscr B}=\{{\mathcal B}({\cal I})\}_{{\cal I} \in {\mathscr I}}$, where ${\mathscr I}$ is the collection of all information sets ${\cal I}$ and ${\mathcal B}({\cal I}) \in \Delta (\Phi)$ is the belief at ${\cal I}$. Assessment ${\cal A} = \{\sigma, {\mathscr B} \}$ is consistent if the system of beliefs is generated by Bayesian updating, whenever possible, given the common prior $\mu$ and the strategy $\sigma$.

\begin{Definition}\label{continuation payoff}
The continuation payoff of player $a_k = i$ at time $t_k$ and state $\phi$, given strategy profile $\sigma$, information set ${\cal I}_k$, and system of beliefs ${{\mathscr B}}$, is 

\[V_i({\cal I}_k,\phi,\sigma, {{\mathscr B}})= \]
\[E_{{\mathcal B}({\cal I}_k)}\Bigg[\sum_{m=0}^{\infty}\beta^{nm}\bigg(s \Big (y_{k+nm} (\sigma, \phi | {\cal I}_k), X(\phi) \Big) - s \Big (y_{k+nm-1} (\sigma, \phi | {\cal I}_k), X(\phi) \Big) \bigg) -cK \Big ({\cal R}_{k+nm}(\sigma,\phi | {\cal I}_k) \Big) \bigg) \Bigg].\]
\end{Definition}

We fix the state $\phi$ for the remainder of the analysis and omit it to simplify the notation.


\begin{Definition}\label{Perfect Bayesian eq def}
Consistent assessment $(\sigma^*, {\mathscr B})$ is  a Perfect Bayesian equilibrium if there is no information set ${\cal I}_k$, player $a_k =i$ and alternative strategy $\sigma = (\sigma_i, \sigma^*_{-i})$, 
such that
\[V_i({\cal I}_k, \sigma, {\mathscr B})  > V_i({\cal I}_k, \sigma^*, {\mathscr B}).\] 
\end{Definition}

The existence of a Perfect Bayesian equilibrium in every game $\Gamma^{k}$ with finitely many actions and periods follows from standard results, as it is weaker than a sequential equilibrium which always exists (\cite{krepsWilson1982}). 

%
%


%

\subsection{Truncated Games with Finitely Many Periods}

For each game $\Gamma^\infty$ with infinitely many periods and finitely many actions, we will generate a sequence of truncated games with finitely many periods. Following \cite{fudenbergLevine83, fudenbergLevine86}, we assume, without loss of generality, that there is a ``do nothing'' or null action at each time $t$, denoted $0$. The $0$ action means that the Trader repeats, with probability 1, the previous announcement, and buys no signals. Note that the null action guarantees that the agent's payoff at that period is 0, irrespective of the previous announcement. With each game $\Gamma^\infty$, we associate a collection of truncated games. A truncated game $\Gamma(m)$  effectively ends in time $t_m$, because for all $t_k>t_m$, the only available action for any Trader is the null action of repeating the previous announcement and their beliefs no longer update.  Hence, the announcement in $t_m$ is repeated by everyone. Note that a game $\Gamma^m$ ends at $t_m$, whereas game $\Gamma(m)$ has infinitely many periods but traders can only choose the null action after $t_m$. 

More formally, let $\Sigma(m)$ be the strategy space of truncated game $\Gamma(m)$. Each strategy profile $\sigma \in \Sigma(m)$ is of the form $(\sigma_1, \sigma_2, \ldots, \sigma_m, 0, 0, \ldots)$, where $\sigma_k$ maps information sets at time $t_k$ to  the set of possible (finite) signals and announcements ${\cal Y}$, for the Trader who makes an announcement at $t_k$. Note that $\Sigma(1) \sq \Sigma(2) \sq \ldots \sq \Sigma(\infty)$, where $\Sigma(\infty) = \Sigma$ is the strategy space of the infinite horizon game. Moreover, each $\Sigma(m)$ is compact.

Let ${\cal A}(m)$ be the collection of assessments ${\cal A} = (h, {\mathscr B})$ such that $h \in \Sigma(m)$ and
beliefs do not update after $m$, so that if  $k>m$ and information set ${\cal I}_{k-1}$ immediately precedes ${\cal I}_k$, then ${\mathcal B}({\cal I}_k) = {\mathcal B}({\cal I}_{k-1})$. Let ${\mathscr B}(m)$ be the collection of all systems of beliefs ${\mathscr B}$ of game $\Gamma(m)$ and note that ${\mathscr B}(m)$ is compact.

An assessment ${\cal A} = (h, {\mathscr B})$ is an equilibrium in $\Gamma(m)$  only if ${\cal A} = (h, {\mathscr B}) \in {\cal A}(m)$. Note that ${\cal A}(1) \sq {\cal A}(2) \sq \ldots \sq {\cal A}(\infty)$, where ${\cal A}(\infty)$ is the collection of all assessments in the infinite horizon game $\Gamma^\infty$. Intuitively, after $t_m$ no Trader changes her action, hence there is no information revelation and beliefs do not update. Set ${\cal A}(\infty)$ is compact because it is a closed subset of $\times_{m=1}^{\infty}(\Sigma(m) \times {\mathscr B}(m))$, which is the product of compact sets in the product topology, and therefore also compact.



\subsection{Continuity}

In order to approximate an equilibrium in the infinite game $\Gamma^\infty$ we need to define the distance between assessments. Recall that $\Phi$ is the space of all uncertainty and let $\delta(p,q) = \underset{E \in \Phi}{\sup}|p(E)-q(E)|$ be the distance between two beliefs $p,q \in \Delta(\Phi)$.


We define the distance between two systems of beliefs ${\mathscr B},{\mathscr B}'$ as
\[d({\mathscr B},{\mathscr B}') = \underset{{\cal I}_k \in {\mathscr I}}{\sup} \left \{\frac{1}{k}  \delta \left ({\mathcal B}({\cal I}_k),{\mathcal B}'({\cal I}_k) \right) \right \}.\]

This metric is motivated in \cite{fudenbergLevine83}. It specifies that two systems of beliefs are close to each other if they only differ in the distant future. We can similarly measure the distance between two 
strategies  $\sigma, \sigma'$ with $d(\sigma, \sigma') = \underset{{\cal I}_k \in {\mathscr I}}{\sup} \left \{\frac{1}{k} \delta \left (\sigma({{\cal I}_k}),\sigma'({{\cal I}_k}) \right ) \right \}$, where $\sigma({{\cal I}_k}) \in \Delta \left ( \underset{{\cal R} \in {\cal E}}{\bigcup}{\cal Y}^{T_{\cal R}} \right )$ is the strategy of the Trader who announces at $t_k$ and information set ${\cal I}_k$. Finally, we can also measure the distance between truncated systems of beliefs (and strategies) that occur after an information set. 
Given an information set ${\cal I}_k$, let ${\cal I}_k({\mathscr B})$ be the restriction of ${\mathscr B}$ to all information sets that succeed ${\cal I}_k$, including ${\cal I}_k$. The same notation ${\cal I}_k(\sigma)$ applies to a strategy $\sigma$. Then, the distances $d({\cal I}_k({\mathscr B}), {\cal I}_k({\mathscr B}'))$  and $d({\cal I}_k(\sigma), {\cal I}_k(\sigma'))$ calculate the distance only with respect to information sets that follow ${\cal I}_k$.

We can now define the  distance between two assessments ${\cal A} = \{\sigma, {\mathscr B} \}, {\cal A}' = \{\sigma', {\mathscr B}' \}$:

\[d({\cal A}, {\cal A}') \equiv \underset{{\cal I}_k \in {\mathscr I}}{\sup} \left \{d({\cal I}_k(\sigma), {\cal I}_k(\sigma')), \left \{ \underset{h_{a_k} \in \Sigma_{a_k}}{\sup} d({\cal I}_k(h_{a_k}, \sigma_{-{a_k}}), {\cal I}_k(h_{a_k}, \sigma'_{-{a_k}})) \right \}, d({\cal I}_k({\mathscr B}), {\cal I}_k({\mathscr B}')) \right \},\]
where $a_k$ is the announcer at time $t_k$ and information set ${\cal I}_k$. \


This metric is also motivated in \cite{fudenbergLevine83}. Two assessments ${\cal A}$ and ${\cal A}'$ are close to each other if the following three conditions are true at each information set ${\cal I}_k$. First, the distributions over actions that are generated by $\sigma$ and $\sigma'$ for every subsequent information set  are close to each other. Second, the distributions over actions are close to each other even when Trader $a_k$ deviates with $h_{a_{k}}$  at ${\cal I}_k$. Finally, the sequence of sets of beliefs that are generated from $ {\mathscr B},  {\mathscr B}'$ given ${\cal I}_k$ are also close to each other. 

%
%
%
%
%

Given an assessment  ${\cal A} = \{\sigma, {\mathscr B} \}$, let $V_i ({\cal I}, \sigma, {\mathscr B})$ be \is continuation payoff at information set ${\cal I}$. We say that the game is uniformly continuous if whenever two assessments are close to each other, the continuation payoffs are also close, for each information set.

\begin{Definition}
The game $\Gamma^\infty$  is uniformly continuous if for all information sets ${\cal I}$ and all sequences of assessments $\{{\cal A}^n\} = \{\sigma^n, {\mathscr B}^n\}$, $\{{\cal A}'^n\}=\{\sigma'^n, {\mathscr B}'^n \}$, ${\cal A}^n \rightarrow {\cal A}'^n$ implies $\big | V_i ({\cal I}, \sigma^n, {\mathscr B}^n) - V_i ({\cal I}, \sigma'^n, {\mathscr B}'^n) \big | \rightarrow 0$ for all $i \in I$. 
\end{Definition}

\cite{ostrovsky12} generates a uniformly continuous game by shortening the time period $t_m$ as $m \rightarrow \infty$. For example, in his version of the \cite{kyle85} model, he sets $t_k = 1-\frac{1}{2^k}$. Although he does not specify how he achieves uniform continuity in the prediction market model, we use something similar. Following \cite{dimitrovSami08}, we assume that the payment at $t_m$ is $\beta^m (s(y_{t_m},x^*)-s(y_{t_{m-1}},x^*))$, instead of $s(y_{t_m},x^*)-s(y_{t_{m-1}},x^*)$.  This is equivalent to shortening the  period $t_m$ as $m \rightarrow \infty$, so that the discounting factor needs to decrease accordingly. Although the results in the main text are true with both definitions, the former ensures that the game is uniformly continuous. We refer to this uniformly continuous game with finitely many announcements as $\Gamma^\infty$. 

Let ${\mathscr B}(\infty)$ be the collection of all systems of beliefs. Let constant $w^m$ be the greatest variation in any agent's payoff and for any system of beliefs, from strategies $\sigma, \tau$ that are identical for all information sets up to time $t_{m-1}$, written as $\sigma =_{m-1} \tau$.
\[w^m \equiv \underset{{\mathscr B} \in {\mathscr B}(\infty)}{\underset{{\cal I} \in {\mathscr I}}{\sup}} \; \underset{\underset{\sigma =_{m-1} \tau}{i \in I}}{\sup}
  \bigg | V_i ({\cal I}, \sigma, {\mathscr B}) - V_i ({\cal I}, \tau, {\mathscr B}) \bigg |.\]
  
\begin{Definition}
The game $\Gamma^\infty$ is continuous at infinity if $w^m \rightarrow 0$ as $m \rightarrow \infty$.
\end{Definition}

\noi It is straightforward that because $\Gamma^\infty$  is uniformly continuous, it is also continuous at infinity.

\subsection{Existence of Equilibrium in $\Gamma^\infty$}

Following \cite{fudenbergLevine83, fudenbergLevine86}, we prove existence in the infinite game $\Gamma^\infty$ through a series of lemmas. We first define the notion of an $\epsilon$-Perfect Bayesian equilibrium.


\begin{Definition}\label{e Perfect Bayesian eq def}
Consistent assessment $(\sigma^*, {\mathscr B})$ is  an $\epsilon$-Perfect Bayesian equilibrium if there is no information set ${\cal I}_k$, player $a_k =i$ and alternative strategy $\sigma = (\sigma_i, \sigma^*_{-i})$, 
such that
\[V_i({\cal I}_k, \sigma, {\mathscr B})  > V_i({\cal I}_k, \sigma^*, {\mathscr B}) + \epsilon.\] 
\end{Definition}

\begin{Lemma} 
\label{lemma 2.1}
If $(h^*, {\mathscr B}) \in {\cal A}(m)$ is an $\epsilon$-Perfect Bayesian equilibrium in $\Gamma(m)$, then $(h^*, {\mathscr B})$ is an $(\epsilon+w^m)$-Perfect Bayesian equilibrium in $\Gamma^\infty$.
\end{Lemma}

\begin{proof} 
 Suppose $(h^*, {\mathscr B})$ is an $\epsilon$-Perfect Bayesian equilibrium in $\Gamma(m)$  and let $g \in \Sigma(\infty)$. Set $h = (g_1, g_2, \ldots, g_m, 0, \ldots)$. 
 For any information set ${\cal I} \in {\mathscr I}$ where agent $i$ makes an announcement, we have
\[V_i ({\cal I}, h_i, h^*_{-i}, {\mathscr B}) - V_i ({\cal I}, h^*, {\mathscr B}) \leq \epsilon. \]
Because $h$ and $g$ differ only after $t_m$, we have
\[V_i ({\cal I}, g_i, h^*_{-i}, {\mathscr B}) - V_i ({\cal I}, h_i, h^*_{-i}, {\mathscr B}) \leq w^m. \]
Adding the two inequalities we have
\[V_i ({\cal I}, g_i, h^*_{-i}, {\mathscr B}) - V_i ({\cal I}, h^*, {\mathscr B}) \leq  \epsilon + w^m. \]



For any information set ${\cal I}_k$ where $t_k \leq t_m$, the consistency on beliefs is satisfied because $(h^*, {\mathscr B})$ is an $\epsilon$-Perfect Bayesian equilibrium in $\Gamma(m)$. For any other information set, the consistency on beliefs is also satisfied because everyone chooses the null action of repeating the previous announcement, hence there is no updating of information or beliefs. Because $g$ was arbitrary, $(h^*, {\mathscr B})$ is an $(\epsilon+w^m)$-Perfect Bayesian equilibrium in $\Gamma^\infty$.

\end{proof}

\begin{Lemma}
\label{lemma 3.2}
Consider a sequence of assessments $\{{\cal A}_m\} = \{g_m, {\mathscr B}_m\}$ such that each ${\cal A}_m$ is an $\epsilon$-Perfect Bayesian equilibrium in $\Gamma^\infty$ and ${\cal A}_m \rightarrow {\cal A} = \{g, {\mathscr B}\}$. Then, ${\cal A}$ is an $\epsilon$-Perfect Bayesian equilibrium in $\Gamma^\infty$.
\end{Lemma}

\begin{proof}
First, note that since ${\cal A}_m \rightarrow {\cal A} = (g, {\mathscr B})$, the corresponding sets of beliefs ${{\mathscr B}}_m$ converge to ${{\mathscr B}}$, hence the consistency condition on beliefs is satisfied. Suppose there is an information set ${\cal I}$ and a strategy $h_i$ such that 
\[V_i ({\cal I}, h_i, g_{-i}, {\mathscr B}) - V_i ({\cal I}, g, {\mathscr B}) \geq \epsilon + 3 \delta. \]

Because ${\cal A}_m \rightarrow {\cal A}$ and the game is uniformly continuous, for any $\delta$, there exists large $m$ such that 
\[V_i ({\cal I}, g_m, {\mathscr B}_m) - V_i ({\cal I}, g, {\mathscr B}) < \delta, \]
\[V_i ({\cal I}, h_i, g_{-i}, {\mathscr B}) - V_i ({\cal I}, h_i, {g_m}_{-i}, {\mathscr B}_m) < \delta. \]

Combining the three inequalities, we have

\[V_i ({\cal I}, h_i, g_{-i}, {\mathscr B}) - V_i ({\cal I}, g, {\mathscr B}) \geq \epsilon + 3 \delta > \epsilon + \delta + V_i ({\cal I}, g_m, {\mathscr B}_m) - V_i ({\cal I}, g, {\mathscr B})  +  \]
\[V_i ({\cal I}, h_i, g_{-i}, {\mathscr B}) - V_i ({\cal I}, h_i, {g_m}_{-i}, {\mathscr B}_m), \]
\noindent which implies

\[V_i ({\cal I}, h_i, {g_m}_{-i}, {\mathscr B}_m)   -  V_i ({\cal I}, g_m, {\mathscr B}_m)> \epsilon + \delta. \]

As $\delta$ can be taken to be arbitrarily small, we can find big enough $m$ so that 
\[V_i ({\cal I}, h_i, {g_m}_{-i}, {\mathscr B}_m)   -  V_i ({\cal I}, g_m, {\mathscr B}_m)> \epsilon, \]
which contradicts that ${\cal A}_m$ is an $\epsilon$-Perfect Bayesian equilibrium.
\end{proof}

\begin{Lemma}
\label{limit theorem}
Suppose that there is a sequence $\{{\cal A}_m\} = \{g_m, {\mathscr B}_m\}$  such that ${\cal A}_m$ is an $\epsilon_m$-Perfect Bayesian  equilibrium in $\Gamma(m)$ and, as $m \rightarrow \infty$, we have $\epsilon_m \rightarrow 0$ and ${\cal A}_m \rightarrow {\cal A}^* = (g^*, {\mathscr B}^*)$. Then, ${\cal A}^*$ is a Perfect Bayesian  equilibrium in $\Gamma^\infty$.
\end{Lemma}

\begin{proof}

From Lemma \ref{lemma 2.1}, ${\cal A}_m$ is an $(\epsilon_m + w^m)$-Perfect Bayesian equilibrium in the infinite game $\Gamma^\infty$. Because $\Gamma^\infty$ is uniformly continuous, it is also continuous at infinity and, together with $\epsilon_m \rightarrow 0$, we have  $\epsilon_m + w^m \rightarrow 0$. This implies that for each $\delta > 0$ there is $M$ such that $\epsilon_m + w^m < \delta$, whenever $m > M$. From Lemma \ref{lemma 3.2}, ${\cal A}^*$ is a $\delta$-Perfect Bayesian equilibrium in the infinite game $\Gamma^\infty$. Since this is true for every $\delta >0$, ${\cal A}^*$  is a Perfect Bayesian equilibrium in $\Gamma^\infty$.
\end{proof}

\begin{proposition} \label{corollary 4.2}
There is a Perfect Bayesian equilibrium in $\Gamma^\infty$.
\end{proposition}

\begin{proof}

Each finite-horizon $\Gamma^m$ and therefore $\Gamma(m)$  has a Perfect Bayesian equilibrium ${\cal A}_m$, for any $m < \infty$. From Lemma \ref{lemma 2.1}, ${\cal A}_m$ is a $w^m$-Perfect Bayesian equilibrium in $\Gamma^\infty$. Since ${\cal A}(\infty)$  is compact, 
there is a subsequence ${\cal A}_k \sq {\cal A}_m$ with ${\cal A}_k \rightarrow {\cal A}^*$. From Lemma \ref{limit theorem}, ${\cal A}^*$ is a Perfect Bayesian equilibrium in $\Gamma^\infty$.
\end{proof}

\subsection{Convergence of Announcements}

When ${\cal Y}$ is finite, the announcements in equilibrium ${\cal A}^*$ do not necessarily converge to the true value $X(\omega)$, for all states, because $X(\omega)$ may not belong to  ${\cal Y}$. We argue here that, nevertheless, they are as close as possible.%
\footnote{Note also that, as shown in Example \ref{main example}, signals with finitely many realizations can generate $\kappa$ separable securities.}

Proposition \ref{corollary 4.2} shows that the subsequence of Perfect Bayesian equilibria ${\cal A}_k = \{g_k, {\mathcal B}_k\}$,  one for each truncated game $\Gamma(k)$, converges to a Perfect Bayesian equilibrium ${\cal A}^* = (g^*, {\mathscr B}^*)$ in $\Gamma^\infty$. Given that each $\Gamma(k)$ has effectively finitely many periods, the last Trader to announce makes the myopic best announcement, given the available information. Moreover, a proper scoring rule is `order-sensitive' so that the further away the forecast is from the true expected value, the lower is the expectation of the score (see footnote \ref{fn: ostrovsky p2618}). This implies that the myopic strategy in the last period is to announce, with probability one, the element in ${\cal Y}$ which is closest to the expected value of $X$. Given that ${\cal A}_k \rightarrow {\cal A}^*$, so that both beliefs and strategies converge to ${\cal A}^*$, we have that the sequence of announcements in ${\cal A}^*$ converge to a limit which is as close as possible to the true expected value of $X$, given the limit of the beliefs in ${\mathscr B}^*$.%
\footnote{The proof of Theorem 1 in \cite{ostrovsky12}, which we use for the proof of our Theorem \ref{thm:info aggregation}, shows that in all equilibria, the updating of beliefs is a martingale, hence through the martingale convergence theorem the beliefs of all traders converge to a limit which incorporates all available information, given the equilibrium strategies.}
Finally, as we increase the number of announcements which are evenly distributed in ${\cal Y}$, the distance between the limit of announcements and the true value $X(\omega)$ decreases.

\section{Proofs}

\printProofs


\pagebreak

\bibliographystyle{plainnat}
\bibliography{biblio}


\end{document}